\documentclass[prodmode,acmtomacs]{acmsmall} 

\usepackage[ruled]{algorithm2e}
\usepackage{amsmath}
\usepackage{graphicx}
\usepackage{epstopdf}
\usepackage{dsfont}
\usepackage{color}
\usepackage{bm}

\SetAlFnt{\small}
\SetAlCapFnt{\small}
\SetAlCapNameFnt{\small}
\SetAlCapHSkip{0pt}
\IncMargin{-\parindent}

\acmVolume{0}
\acmNumber{0}
\acmArticle{0}
\acmYear{2017}
\acmMonth{0}

\newtheorem{assumption}[theorem]{Assumption}

\numberwithin{equation}{section}

\newcommand{\red}[1]{{\color{black}#1}}
\newcommand{\abs}[1]{\left|#1\right|}
\allowdisplaybreaks

\doi{0000001.0000001}

\issn{1234-56789}

\begin{document}

\markboth{H. Zhu,T. Liu and E. Zhou}{Risk Quantification in Stochastic Simulation under Input Uncertainty}

\title{Risk Quantification in Stochastic Simulation under Input Uncertainty}
\author{HELIN ZHU
\affil{Georgia Institute of Technology}
TIANYI LIU
\affil{Georgia Institute of Technology}
ENLU ZHOU
\affil{Georgia Institute of Technology}
}

\begin{abstract}
When simulating a complex stochastic system, the behavior of output response depends on input parameters estimated from finite real-world data, and the finiteness of data brings input uncertainty into the system.
The quantification of the impact of input uncertainty on output response has been extensively studied.
Most of the existing literature focuses on providing inferences on the mean response at the true but unknown input parameter, including point estimation and confidence interval construction.
Risk quantification of mean response under input uncertainty often plays an important role in system evaluation and control, because it provides inferences on extreme scenarios of mean response in all possible input models.
To the best of our knowledge, it has rarely been systematically studied in the literature.
In this paper, first we introduce risk measures of mean response under input uncertainty, and propose a nested Monte Carlo simulation approach to estimate them.
Then we develop asymptotical properties such as consistency and asymptotic normality for the proposed nested risk estimators.
 We further study the associated budget allocation problem for efficient nested risk simulation, and \red{finally use a sharing economy example to illustrate the importance of accessing and controlling risk due to input uncertainty.}
\end{abstract}

%
%

%
%



\acmformat{Helin Zhu, Tianyi Liu and Enlu Zhou, 2017. Risk Quantification in Stochastic Simulation under Input Uncertainty.}

\begin{bottomstuff}
This work was supported by National Science Foundation under Grants CMMI-1413790 and CAREER CMMI-1453934, and Air Force Office of Scientific Research under Grant YIP FA-9550-14-1-0059.

Author's addresses: H. Zhu, T. Liu {and} E. Zhou, H. Milton Stewart School of Industrial and Systems Engineering, Georgia Institute of Technology, Atlanta, GA 30332.
\end{bottomstuff}

\maketitle

\section{Introduction and Motivation}\label{Sec1:Introduction}
For a complex real-world stochastic system, simulation is a powerful tool to analyze its behavior when real experiments on the system are expensive or difficult to conduct. Simulation is driven by input models that are distributions capturing the randomness in the system. For example, when simulating a queueing network, the random customer arrival and service times are generated from appropriate distributions (i.e., input models).
The uncertainty on input parameters (e.g., customer arrival rates and service rates) may need to be taken into account, since they are typically estimated from finite records of historical data.
In general, there are two sources of uncertainty in a typical stochastic simulation experiment: the extrinsic uncertainty on input parameters (referred to as \emph{input parameter uncertainty}, or simply \emph{input uncertainty}) that reflects the variability of the finite data used to estimate input parameters,
and the intrinsic uncertainty on output response (referred to as \emph{stochastic uncertainty}) that reflects the inherent stochasticity of the system.

The variability of simulation output response clearly depends on both input uncertainty and stochastic uncertainty.
An important question to address is how to quantify the impact of input uncertainty on output response variability in the presence of stochastic uncertainty.
Various quantification methods have been proposed, including frequentist and Bayesian methods among many others.
Frequentist methods include the Direct/Bootstrap Resampling methods by \cite{barton1993uniform,barton2001resampling}, \cite{cheng1997sensitivity}, etc.
The input model for these methods can be a non-parametric empirical distribution or a parametric distribution estimated from historical data.
Bayesian methods include the Bayesian Model Averaging (BMA) methods by \cite{chick2001input}, \cite{zouaoui2003accounting,zouaoui2004accounting}, \cite{biller2011accounting}, etc.
In these methods, a Bayesian updating rule is applied on a chosen prior distribution of input parameter to generate a posterior parameter distribution, which is then used as the sampling distribution of input parameter in the simulation experiment.
In addition to these methods, \cite{cheng1997sensitivity} also develops the $\delta$-method, which decomposes the variance of output response into two components that are caused by input uncertainty and stochastic uncertainty, respectively.
\cite{song2015quickly} develops a method for quickly assessing the relative contribution of each input distribution to the overall variance.
In recent years, with the rise of stochastic kriging in stochastic simulation (e.g., \cite{ankenman2010stochastic}), meta-model assisted methods have been developed for quantifying input uncertainty, see \cite{barton2013quantifying}, \cite{xie2014bayesian,xie2014statistical}, etc.
\cite{henderson2003} provides an early review on the importance of input uncertainty and common methods to deal with it. \cite{barton2012tutorial} provides a more recent review on popular methods in output analysis under input uncertainty, and highlights some remaining challenges in this area.

Some of the aforementioned works aim at providing inferences on the mean response at the true but unknown input parameter, often through point estimation and confidence interval (CI) construction.
Some others focus on obtaining an empirical distribution of mean response, and providing a more complete picture of all possible scenarios of mean response under input uncertainty.
However, to the best of our knowledge, rigorous quantification of extreme scenarios of mean response in all possible input models is still lacking.
Such quantification could provide inferences on system sensitivity or robustness to input uncertainty, and thus would be critical for control of the system especially when the decisions being made are irrevocable.

\textcolor{black}{For example, consider the system of a call center.  There is often a service-level agreement in which there are penalties based on whether the fraction of calls answered within a certain time period exceeds a threshold.  
After the number of operators being determined,   we can hardly change this decision and ask more operators to work temporarily. In this situation, it is critical to assess and control the risk of low fraction of calls answered, in all possible input models, since this might lead to penalties which will affect the reputation of the company.}

For another example, consider a large-scale power system. It is usually too expensive or risky to conduct real experiments on the system operation, and therefore, stochastic simulation is often used to study the economics, reliability, and emission variable effects of power systems operating in a market environment (\cite{degeilhstochastic}).
In a typical power system simulation experiment, the inputs may include the resource parameters, the loading (market demand) parameters, etc., which all exhibit variability and uncertainty.
The risk quantification and management of system performance under input uncertainty is of great importance because extreme scenarios of mean response (e.g., high mean power loads during peak time) might cause a part or whole breakdown of the power system and lead to disastrous outcomes.


In this paper, we aim to quantify the risk in stochastic simulation under input uncertainty, by studying risk measures of mean response w.r.t. the distribution of input parameter.
We will focus on risk measures such as Value-at-Risk (VaR) and Conditional Value-at-Risk (CVaR).
Loosely speaking, VaR characterizes an extreme (e.g., 99\%) quantile of the mean response distribution, and CVaR characterizes the conditional expectation of a tail portion of the mean response distribution.
VaR, as one of the very earliest risk measures introduced in financial risk management, is easy to understand and interpret for practitioners. CVaR, as a classic coherent risk measure (see, e.g., \cite{artzner1997coherent}), exhibits nice properties such as convexity and monotonicity for optimization (see, e.g., \cite{rockafellar2000optimization}).
They have been extensively used in financial industry, especially after the financial crisis in 2008. An abundant literature has dedicated to studying the estimation and optimization of risk measures under various settings; in particular, \cite{hong2014monte} provides a comprehensive review of Monte Carlo methods for VaR and CVaR.

We will introduce VaR and CVaR for quantifying the risk in stochastic simulation under input uncertainty, and provide numerical schemes for their estimation.
Specifically, we will study nested Monte Carlo estimators for VaR and CVaR of mean response from both theoretical and computational perspectives. Our numerical examples illustrate the importance and necessity of risk quantification under input uncertainty.
To summarize, the contributions of this paper are three-fold:
\begin{itemize}
\item[(1)] For output analysis in stochastic simulation, our work is among the first to systematically study risk quantification of mean response in all possible input models using risk measures.

\item[(2)] Under the respective ``Weak Assumption'' and ``Strong Assumption'' (elaborated in Section \ref{Sec3:Asymptotic}), we show that the proposed nested risk estimators are consistent in different limiting senses. \textcolor{black}{Under ``Strong Assumption", they are  asymptotically normally distributed as well, which is the guarantee for constructing asymptotically valid CIs.}

\item[(3)] We solve the associated budget allocation problem that arises in nested simulation of risk estimators, in order to improve simulation efficiency. The numerical study demonstrates the effectiveness of our approach and shows that the obtained budget allocation schemes drastically reduce the widths of the CIs constructed.
\end{itemize}

We note that, in a broader sense, our framework bears some similarity with risk assessment in credit management, since both of them deal with simulating certain conditional expectations.
The work most relevant to ours is probably \cite{gordy2010nested}, in which the authors study the asymptotic representation of the Mean Squared Error (MSE) of nested risk estimators in credit risk management.
By minimizing MSE asymptotically, they obtain an (asymptotically) optimal budget allocation scheme.
In contrast, our work focuses on the analysis of asymptotical properties such as consistency and asymptotic normality of the proposed nested risk estimators.
\textcolor{black}{Furthermore, the associated budget allocation problem in our approach is to minimize the widths of the wider half CIs, which is similar to the MSE criterion in \cite{gordy2010nested} but from a different point of view. Moreover, we propose a new approach to estimate all the parameters needed in the problem, so that this budget allocation procedure can be used widely in practice.}

Other common approaches for credit risk management include but not limited to the delta-gamma method by \cite{rouvinez1997}, \cite{glasserman2000variance}, etc;
the two-level confidence interval procedure with screening by \cite{lan2010confidence}, etc;
the stochastic kriging method by \cite{liu2010stochastic}, etc;
the ranking and selection method by \cite{broadie2011efficient}, etc.
Among other relevant literature, \cite{lee1998monte} studies point estimation of a quantile (VaR) of the distribution of a conditional expectation via a two-level simulation;
\cite{steckley2006} considers estimating the density of a conditional expectation using kernel density estimation;
\cite{sun2011efficient} studies efficient nested simulation for estimating the variance of a conditional expectation.
Most of these works focus on efficient allocation of inner simulation sizes across different outer scenarios, and \cite{lee1998monte}, \cite{steckley2006}, and \cite{sun2011efficient} consider optimal allocation between inner and outer sampling.
Our work distinguishes from these works in that we focus on the theoretical properties of nested risk estimators, and our budget allocation scheme can be viewed as a byproduct of the theoretical properties established.
We do point out that varying inner-layer sample sizes across different outer-layer scenarios, as studied in some of the aforementioned works, could be further incorporated here to improve simulation efficiency;
however, it is beyond the scope of this paper.

The rest of the paper is organized as follows.
In Section \ref{Sec2:Risk}, we introduce risk measures VaR and CVaR of mean response w.r.t. input uncertainty, and propose nested risk estimators for risk quantification in stochastic simulation under input uncertainty.
In Section \ref{Sec3:Asymptotic}, we establish the asymptotical properties of the proposed nested risk estimators under different assumptions, and then construct asymptotically valid CIs.
We formulate the associated budget allocation problem and propose a new approach to solve it in Section \ref{Sec4:Budget}.
In Section \ref{Sec5:Numerical}, we conduct numerical experiments to demonstrate some of the theoretical results from previous sections.
Conclusions are provided in Section \ref{Sec6:Consusion}.

\section{Risk Measures of Mean Response under Input Uncertainty}\label{Sec2:Risk}

\subsection{Formulation}\label{Sec2.sub1:Formulation}

Let us first rigorously define risk measures VaR and CVaR of mean response under input uncertainty.
In a stochastic simulation experiment, consider a response function in the form of $h(\theta;\xi)$, where $\theta$ represents the input parameter(s) and $\xi$ represents the noise (stochastic uncertainty) in the response.
Let $H(\theta)=\mathbb{E}_{\xi}[h(\theta;\xi)]$ be the mean response,
and thus $h(\theta;\xi)=H(\theta)+\mathcal{E}(\theta;\xi)$,
where $\mathcal{E}(\theta;\xi)$ is the stochastic noise that satisfies $\mathbb{E}[\mathcal{E}(\theta;\xi)|\theta]=0$ and $Var[\mathcal{E}(\theta;\xi)|\theta]=\tau^2_\theta$.
Here assume $\tau^2_\theta$ is a finite deterministic function of $\theta$.
Furthermore, suppose there is a probability distribution (called ``belief distribution'') on $\theta$ that reflects our belief on input uncertainty, since $\theta$ needs to be inferred from finite historical data.
For example, if one takes a Bayesian approach, then the belief distribution is constructed via Bayesian updating. Of course, there are other approaches such as bootstrapping.
Specifically, suppose $p_o(\theta)$ is a prior distribution on $\theta$, and it could be either non-informative or informative depending on prior knowledge.
Then the posterior distribution $p(\cdot|\mathbf{x})$ is obtained via sequential Bayesian updating with historical data $\mathbf{x}$.
Assume $\tau^2 :=\int \tau^2_\theta p(\theta|\mathbf{x})d\theta$ is also finite.

Let $0<\alpha<1$ be the risk level of interest (e.g., $\alpha=0.99$).
Then VaR of the mean response $H(\theta)$, denoted by $v_{\alpha}\left(\mathbb{E}_{\xi}[h(\theta;\xi)]\right)$ (or interchangeably $v_{\alpha}\left(H(\theta)\right)$), is defined by the $\alpha$-quantile of $H(\theta)$, i.e.,
\begin{equation}\label{eq.3.1}
v_{\alpha}\left(H(\theta)\right)\overset{\triangle}=\inf\{t: F(t)\ge \alpha\},
\end{equation}
where $F(\cdot)$ is the cumulative distribution function (c.d.f.) of $H(\theta)$.
When $H(\theta)$ admits a positive and continuous probability density function (p.d.f.), which is denoted by $f(\cdot)$, around $v_{\alpha}\left(H(\theta)\right)$, (\ref{eq.3.1}) can be simplified as $v_{\alpha}\left(H(\theta)\right)
=F^{-1}(\alpha)$.
CVaR of $H(\theta)$, denoted by $c_{\alpha}\left(\mathbb{E}_{\xi}[h(\theta;\xi)]\right)$ (or interchangeably $c_{\alpha}\left(H(\theta)\right)$), is defined by the conditional expectation of the $\alpha$-tail distribution of $H(\theta)$, i.e.,
\begin{eqnarray}\label{eq.3.2}
c_{\alpha}\left(H(\theta)\right)
\overset{\triangle}=v_{\alpha}(H(\theta))+
\frac{1}{1-\alpha}\mathbb{E}_{p(\cdot|\mathbf{x})}
\left[\left(H(\theta)-v_{\alpha}(H(\theta))\right)^+\right].
\end{eqnarray}
With slight abuse of notations, we use $v_\alpha$ and $c_\alpha$ as the abbreviations for $v_{\alpha}\left(H(\theta)\right)$ and $c_{\alpha}\left(H(\theta)\right)$, respectively.

Calculating risk measures such as $v_\alpha$ and $c_\alpha$ is straightforward when the system is simple.
For example, when the p.d.f. of $H(\theta)$ admits an explicit expression, VaR or CVaR of $H(\theta)$ could be calculated via numerical integration.

\subsection{Nested Simulation of VaR and CVaR}\label{Sec2.sub2:Nested}

Let us first consider Monte Carlo estimation of $v_\alpha$ and $c_\alpha$ without the presence of stochastic uncertainty. That is, $H(\theta)$ can be evaluated exactly for all $\theta$.

First, draw $N$ i.i.d. scenarios $\theta_1,...,\theta_N$ from the belief distribution $p(\theta|\mathbf{x})$; then, simulate $\{H(\theta_i): i=1,...,N\}$ and sort them in ascending order, denoted by $H(\theta_{(1)})\le H(\theta_{(2)})\le\cdot\cdot\cdot\le H(\theta_{(N)})$; finally, estimators of $v_\alpha$ and  $c_{\alpha}$ are given, respectively, by
\begin{eqnarray*}
\widehat{v}_{\alpha}&=&H(\theta_{(\alpha N)}),\label{eq.3.3}\\
\widehat{c}_{\alpha}&=&
\widehat{v}_{\alpha}+\frac{1}{(1-\alpha) N}\sum\limits_{i=1}^{N} \left(H(\theta_{i})-\widehat{v}_{\alpha}\right)^+\label{eq.3.4},
\end{eqnarray*}
where for convenience we assume $\alpha N$ is an integer. Intuitively, $\widehat{v}_{\alpha}$ is the $\alpha$-level VaR of the empirical mean response distribution consisting of $\{H(\theta_{(i)}): i=1,...,N\}$. In parallel, $\widehat{c}_{\alpha}$ is the $\alpha$-level CVaR of the empirical mean response distribution. The properties of $\widehat{v}_{\alpha}$ and $\widehat{c}_{\alpha}$ have been well-studied in the literature. For example, although $\widehat{v}_\alpha$ and $\widehat{c}_{\alpha}$ are not unbiased, they are strongly consistent and asymptotically normally distributed under appropriate regularity conditions (\cite{sun2010asymptotic}).

When stochastic uncertainty is present, the exact value of $H(\theta)$ might not be readily available; instead, it is estimated via sample average. Naturally, to obtain estimators of $v_\alpha$ and $c_\alpha$, we can extend the estimation procedure described above by replacing $\{H(\theta_i)\}$ with their sample average estimates $\{\widehat{H}(\theta_i)\}$.
Specifically, for each input scenario $\theta_i$, simulate $M$ i.i.d. response samples $\{h(\theta_i; \xi_{ij}): j=1,...,M\}$; then, approximate $H(\theta_i)$ by $\widehat{H}_M(\theta_i)=\frac{1}{M}\sum_{j=1}^{M} h(\theta_i;\xi_{ij})$ and sort them in ascending order, denoted by $\widehat{H}_M(\theta^{(1)})\le \widehat{H}_M(\theta^{(2)})\le \cdot\cdot\cdot \le \widehat{H}_M(\theta^{(N)})$; finally, estimate $v_\alpha$ and $c_\alpha$, respectively, by
\begin{eqnarray}
\widetilde{v}_{\alpha}&=&\widehat{H}_M(\theta^{(\alpha N)}),\label{eq.3.5}\\
\widetilde{c}_{\alpha}&=&
\widetilde{v}_{\alpha}+\frac{1}{(1-\alpha) N}\sum\limits_{i=1}^{N} \left(\widehat{H}_M(\theta_{i})-
\widetilde{v}_{\alpha}\right)^+\label{eq.3.6}.
\end{eqnarray}
We refer to $\widetilde{v}_{\alpha}$ or $\widetilde{c}_{\alpha}$ as ``nested risk estimator'', since nested simulation is incurred in the estimation.
Due to nested simulation, the asymptotical properties of $\widetilde{v}_{\alpha}$ and $\widetilde{c}_{\alpha}$ become more complicated.
In next section, we will show that $\widetilde{v}_{\alpha}$ and $\widetilde{c}_{\alpha}$ maintain to be strongly consistent and asymptotically normally distributed in different limiting senses under different sets of regularity conditions.
Hence, using them as inferences for $v_\alpha$ and $c_\alpha$, respectively, is still reasonable.

Note that the ordered statistics $(\theta^{(1)},...,\theta^{(N)} )$ and $(\theta_{(1)},...,\theta_{(N)} )$ are different.
In fact, for fixed input scenarios $\theta_1,...,\theta_N$, $(\theta_{(1)},...,\theta_{(N)} )$ is a constant vector, while $(\theta^{(1)},...,\theta^{(N)} )$ is a random permutation of $(\theta_{(1)},...,\theta_{(N)} )$ that depends on the realizations of $\{h(\theta_i; \xi_{ij})\}$.

\begin{remark}\label{rem.2.1}
In \cite{barton2013quantifying} and \cite{xie2014bayesian}, the authors use nested VaR estimator $\widetilde{v}_{\rho/2}$ and $\widetilde{v}_{1-\rho/2}$ as the lower-upper boundaries of a credible interval (CrI) for $H(\theta_c)$ with confidence level $(1-\rho)$, where $H(\theta_c)$ is the mean response at the true but unknown input parameter $\theta_c$.
\end{remark}

\section{Asymptotic Analysis of Nested VaR and CVaR Estimators}
\label{Sec3:Asymptotic}

In this section, we analyze the asymptotical properties of nested risk estimators $\widetilde {v}_{\alpha}$ and $\widetilde{c}_{\alpha}$, as the inner and outer sample sizes both go to infinity. In particular, we will prove their strong consistency and asymptotic normality in different limiting senses under different sets of regularity assumptions, which are referred to as ``Weak Assumption'' and ``Strong Assumption'', respectively. \textcolor{black}{Under ``Weak Assumption", the consistency result includes iterative limits which make it hard to use in practice. Thus, the stronger result under ``Strong Assumption" which allows the number of outer layer scenarios and inner layer samples to go to infinity simultaneously will be of help. }

\begin{assumption}\label{asm.3.1} \textbf{Weak Assumption}.
\begin{itemize}
\item [(i)] The response $h(\theta;\xi)$ has finite conditional second moment, i.e., $\tau_\theta^2=\mathbb{E}[h^2(\theta;\xi)|\theta]<\infty$ w.p.1 and $\tau^2=\int \tau_\theta^2 p(\theta|\mathbf{x})d\theta<\infty$.

\item [(ii)] The p.d.f. $f(\cdot)$ of the mean response $H(\theta)$ is positive and continuous, and continuously differentiable around $v_{\alpha}$.
\end{itemize}
\end{assumption}

Assumption \ref{asm.3.1} is weak in the sense that it imposes separate assumptions on input uncertainty and stochastic uncertainty. In contrast to the following Strong Assumption, it does not impose joint assumptions on input uncertainty and stochastic uncertainty.

Notice that
$$
\widehat{H}_M(\theta)= \frac{1}{M}\sum_{j=1}^{M} h(\theta;\xi_{j})=\frac{1}{M}\sum_{j=1}^{M} (H(\theta)+\mathcal{E}(\theta;\xi_{j}))=H(\theta)+
\frac{1}{M}\sum_{j=1}^{M}\mathcal{E}(\theta;\xi_{j}).
$$
Let us define a normalized noise term $\bar{\mathcal{E}}_M$ by
$$
\bar{\mathcal{E}}_M\overset{\triangle}=
\sqrt{M}\cdot\frac{1}{M}\sum_{j=1}^{M}\mathcal{E}(\theta;\xi_{j}).
$$
By Central Limit Theorem, under appropriate assumptions $\bar{\mathcal{E}}_M$ has a limiting distribution as $M\rightarrow \infty$.
Note that $\widehat{H}_M(\theta)=H(\theta)+\bar{\mathcal{E}}_M/\sqrt{M}$, then the effect of the diminishing noise term $\bar{\mathcal{E}}_M/\sqrt{M}$ on the distribution of $\widehat{H}_M(\theta)$ will vanish as $M\rightarrow \infty$.
Therefore, we expect the ``distance'' between the distribution of $\widehat{H}_M(\theta)$ and the distribution of $H(\theta)$ to vanish as $M\rightarrow \infty$.
That is, for a fixed $\theta$, $\widetilde{f}_M \rightarrow f$ as $M \rightarrow \infty$, where $\widetilde{f}_M(\cdot)$ represents the p.d.f. of $\widehat{H}_M(\theta)$.
In particular, the following Strong Assumption guarantees that $\widetilde{f}_M$ converges to $f$ sufficiently fast.

\begin{assumption}\label{asm.3.2} \textbf{Strong Assumption}.
\begin{itemize}
\item [(i)] The response $h(\theta;\xi)$ has finite conditional second moment, i.e., $\tau_\theta^2=\mathbb{E}[h^2(\theta;\xi)|\theta]<\infty$ w.p.1 and $\tau^2=\int \tau_\theta^2 p(\theta|\mathbf{x})d\theta<\infty$.

\item[(ii)] The joint density $p_M(h,e)$ of $H(\theta)$ and $\bar{\mathcal{E}}_M$, and its partial derivatives $\frac{\partial}{\partial h}p_M(h,e)$ and $\frac{\partial^2}{\partial h^2}p_M(h,e)$ exist for each $M$ and for all pairs of $(h, e)$.

\item[(iii)] There exist non-negative functions $g_{0,M}(\cdot)$, $g_{1,M}(\cdot)$ and $g_{2,M}(\cdot)$ such that $p_M(h, e)\le g_{0,M}(e)$, $\abs{\frac{\partial}{\partial h}p_M(h, e)} \le g_{1,M}(e)$, $\abs{\frac{\partial^2}{\partial h^2}g_M(h, e)} \le g_{2,M}(e)$ for all $(h, e)$. Furthermore, $\sup\limits_{M}\int |e|^r g_{i,M}(e) de <\infty$ for $i=0,1,2$, and $0\le r \le 4$.
\end{itemize}
\end{assumption}

Assumption \ref{asm.3.2} is strong in the sense that it imposes joint assumptions on input uncertainty and stochastic uncertainty. In particular, Assumption \ref{asm.3.2}.(i) ensures that $\bar{\mathcal{E}}_M$ has a limiting distribution as $M\rightarrow \infty$; Assumption \ref{asm.3.2}.(ii) and \ref{asm.3.2}.(iii) (similar to Assumption 1 of \cite{gordy2010nested}) ensure that the difference between $\widetilde{f}_M(\cdot)$ and $f(\cdot)$ is of the order $O(\frac{1}{M})$. Assumption \ref{asm.3.2} holds when $h(\cdot,\cdot)$ is sufficiently smooth, and the distributions of $\theta$ and $\xi$ have good structural properties (e.g., finite moments up to some order). Note that when Strong Assumption holds, Weak Assumption naturally holds.

\subsection{Consistency}
\label{Sec3.sub1:Consistency}

It turns out, under Weak Assumption, nested risk estimators $\widetilde {v}_{\alpha}$ and $\widetilde{c}_{\alpha}$ are consistent in the sense that they converge to ${v}_{\alpha}$ and ${c}_{\alpha}$ w.p.1, respectively, when $M$ first goes to infinity and then $N$ goes to infinity. In particular, we have the following Theorem \ref{thm.3.1} on the consistency of $\widetilde {v}_{\alpha}$ and $\widetilde{c}_{\alpha}$ under Weak Assumption.

\begin{theorem}\label{thm.3.1} \textbf{Consistency under Weak Assumption.}
Under Assumption \ref{asm.3.1}, we have
\begin{equation}\label{eq.3.7}
\lim\limits_{N\rightarrow \infty}\lim\limits_{M\rightarrow \infty}
\widetilde{v}_{\alpha}
= v_{\alpha},~ w.p.1, \quad\mbox{and}\quad
\lim\limits_{N\rightarrow \infty}\lim\limits_{M\rightarrow \infty}
\widetilde{c}_{\alpha}
= c_{\alpha},~w.p.1.
\end{equation}
\end{theorem}
\begin{proof}
See Appendix \ref{ap.a}.
\end{proof}

Note that in Theorem \ref{thm.3.1} the limits on $N$ and $M$ are iterated and non-interchangeable.
Intuitively, the inner sample size $M$ going to infinity ensures that, for any fixed $\theta$, $\widehat{H}_M(\theta)\rightarrow H(\theta)$ w.p.1 (by Strong Law of Large Numbers).
It follows that for fixed $\theta_1,...,
\theta_N$,
$(\widehat{H}_M(\theta^{(1)}),...,\widehat{H}_M(\theta^{(N)}))
\rightarrow(H(\theta_{(1)}),...,H(\theta_{(N)}))$ w.p.1. as $M\rightarrow \infty$.
Therefore, $\widetilde{v}_{\alpha}\rightarrow\widehat{v}_{\alpha}$ and $\widetilde{c}_{\alpha}\rightarrow\widehat{c}_{\alpha}$ w.p.1 as $M\rightarrow \infty$.
In view of the fact that $\widehat{v}_{\alpha}\rightarrow v_{\alpha}$ and $\widehat{c}_{\alpha}\rightarrow c_{\alpha}$ w.p.1 as $N\rightarrow\infty$, Theorem \ref{thm.3.1} holds.

When Strong Assumption is imposed, we could strengthen the results in Theorem \ref{thm.3.1}.
In particular, the following Theorem \ref{thm.3.2} shows that the iterated limits on $N$ and $M$ in Theorem \ref{thm.3.1} could be relaxed into simultaneous limits.

\begin{theorem}\label{thm.3.2} \textbf{Consistency under Strong Assumption}.
Under Assumption \ref{asm.3.2}, we have
\begin{equation}\label{eq.3.9}
\lim\limits_{N, M\rightarrow \infty}\widetilde{v}_{\alpha}=v_{\alpha}, ~ w.p.1,
\quad\mbox{and}\quad
\lim\limits_{N, M\rightarrow\infty}\widetilde{c}_{\alpha} =c_{\alpha}, ~ w.p.1.
\end{equation}
\end{theorem}
\begin{proof}
See Appendix \ref{ap.b}.
\end{proof}

Theorem \ref{thm.3.2} implies $\widetilde{v}_{\alpha}$ and $\widetilde{c}_{\alpha}$ converge to ${v}_{\alpha}$ and ${c}_{\alpha}$ w.p.1, respectively, when $N$ and $M$ go to infinity simultaneously.
The intuition is as follows.
For an arbitrary $M$, let us bound the difference between $v_\alpha (\widehat{H}_M(\theta))$ (or $c_\alpha (\widehat{H}_M(\theta))$) and $v_\alpha (H(\theta))$ (or $c_\alpha (H(\theta))$), where note that $v_\alpha (\widehat{H}_M(\theta))$ is VaR of $\widehat{H}_M(\theta)$ and $c_\alpha (\widehat{H}_M(\theta))$ is CVaR of $\widehat{H}_M(\theta)$.
As mentioned previously, Assumption \ref{asm.3.2} ensures that the difference between $\widetilde{f}_M(\cdot)$ and $f(\cdot)$ is of the order $O(\frac{1}{M})$. It follows that the difference between $v_\alpha (\widehat{H}_M(\theta))$ and $v_\alpha (H(\theta))$ is also of the order $O(\frac{1}{M})$. Furthermore, note that $\widetilde{v}_\alpha$  could be regarded as a one-layer estimator of $v_\alpha (\widehat{H}_M(\theta))$, i.e.,
\begin{eqnarray*}
\widetilde{v}_\alpha (H(\theta))=\widehat{v}_\alpha (\widehat{H}_M(\theta)).
\end{eqnarray*}
Under Assumption \ref{asm.3.2}, we could show that $\widehat{v}_\alpha (\widehat{H}_M(\theta))$, i.e., $\widetilde{v}_\alpha (H(\theta))$, converges to $v_\alpha (\widehat{H}_M(\theta))$  w.p.1 uniformly for all $M$ as $N\rightarrow \infty$. Therefore, $\widetilde{v}_\alpha (H(\theta))$ converges to ${v}_\alpha (H(\theta))$ w.p.1 as $N$ and $M$ go to infinity simultaneously. Hence, Theorem \ref{thm.3.2} holds.
\color{black}
\begin{remark}
	 Given the same set of assumptions,  Propositions 2 and 3 in \cite{gordy2010nested} suffice the proof of Theorem 3.4, but we can hardly derive the normality result from them. However, the lemma \ref{lem.b.1}-\ref{lem.b.4} used in our proof can directly lead to Theorem \ref{thm.3.4}. In this respect, our proof has its own value.
\end{remark}
\color{black}
\subsection{Asymptotic Normality and Confidence Intervals}
\label{Sec3.sub2:Normality}
\color{black}
After showing the consistency of $\widetilde{v}_{\alpha}$ and $\widetilde{c}_{\alpha}$, it is natural to consider their asymptotic normality properties and construct the associated CIs.
\textcolor{black}{Since we need iterative limits in  the  consistency result under  ``Weak Assumption", it can hardly be used in practice. So we will only consider  the case under ``Strong Assumption". Following the idea in proving Theorem \ref{thm.3.2}, the error of $\widetilde{v}_\alpha$ (or $\widetilde{c}_\alpha$) is decomposed into two components that respectively account for the one-layer simulation error due to input uncertainty and the simulation bias due to stochastic uncertainty. In particular,}
\begin{equation}\label{eq.3.12}
\color{black}
\widetilde{v}^{N,M}_\alpha -v_{\alpha}=
\left(\widetilde{v}^{N,M}_\alpha-\breve{v}^{M}_\alpha\right)+
\left(\breve{v}^{M}_\alpha -v_{\alpha}\right):=
Err_1+Bias_1
\end{equation}
\textcolor{black}{and}
\begin{equation}\label{eq.3.13}
\color{black}
\widetilde{c}^{N,M}_\alpha -c_{\alpha}=
\left(\widetilde{c}^{N,M}_\alpha-\breve{c}^{M}_\alpha\right)+
\left(\breve{c}^{M}_\alpha -c_{\alpha}\right):=
Err_2+Bias_2.
\end{equation}
\textcolor{black}{where  $\breve{v}_\alpha^{M}=v_\alpha(\widehat{H}_M(\theta))$ and $\breve{c}_\alpha^M=c_\alpha(\widehat{H}_M(\theta)).$ Here $Err_1$ (or $Err_2$) is caused by input uncertainty, and $Bias_1$ (or $Bias_2$) is caused by stochastic uncertainty. By properly choosing $N$ and $M$, we can control the bias. Specifically, we have the following Theorem \ref{thm.3.4} on the asymptotic normality of $\widetilde{v}_\alpha$ and $\widetilde{c}_\alpha$.}

\begin{theorem}\label{thm.3.4}
\textbf{Normality under Strong Assumption}.

Define a function $$\Lambda(t)=1/2 f(t)\mathbb{E}[\tau_\theta^2|H(\theta)=t].$$

Under Assumption \ref{asm.3.2}, the existence of the limit $K^2=\lim\limits_{N, M\rightarrow \infty}N/M^2$ is a sufficient and necessary condition for
\begin{equation}\label{eq.3.24}
\lim\limits_{N,M\rightarrow \infty}\sqrt{N}
\left(\widetilde{v}_{\alpha}-v_{\alpha}\right)
\overset{\mathcal{D}}=\sigma_{v}\mathcal{N}(0,1)+|K|\mu_v
~~\mbox{and}~
\lim\limits_{N,M\rightarrow \infty}\sqrt{N}
\left(\widetilde{c}_{\alpha}-c_{\alpha}\right)
\overset{\mathcal{D}}=\sigma_{c}\mathcal{N}(0,1)+|K|\mu_c,
\end{equation}
where $\sigma_{v} :=\sqrt{\alpha(1-\alpha)}/{f(v_\alpha)},$
 $\sigma_{c} :={\sqrt{Var[\left(H(\theta)
-v_{\alpha}\right)^{+}]}}/{(1-\alpha)}$, $\mu_v=\frac{-\Lambda^\prime(v_{\alpha})}
{f(v_{\alpha})}$ and $\mu_c=\frac{\Lambda(v_{\alpha})}
{(1-\alpha)}$.
\end{theorem}


\begin{proof}
See Appendix \ref{ap.d}.
\end{proof}

Theorem \ref{thm.3.4} is consistent with the results in \cite{gordy2010nested} on the characterizations of the asymptotic variances of $\widetilde{v}_\alpha$ and $\widetilde{c}_\alpha$.
We also note that Theorem \ref{thm.3.4} is stronger in that it directly leads to the results in \cite{gordy2010nested}.
Specifically, by minimizing MSE, \cite{gordy2010nested} shows that the variance and the bias of a nested risk estimator are balanced when the sample size pair $(N,M)$ lives in the regime of $N=O(M^2)$.
Theorem \ref{thm.3.4} here shows a stronger result that a nested risk estimator is asymptotically normally distributed if and only if $(N,M)$ lives in the regime of $N=|K|M^2$.

Following Theorem \ref{thm.3.4}, we can construct CIs for $v_{\alpha}$ and $c_{\alpha}$ of confidence level $(1-\beta)$:
\begin{eqnarray}\label{eq.3.25}
\left[\widetilde{v}_{\alpha}+
\frac{t_{\beta/2,N-1}\widehat{\sigma}_{v}}{\sqrt{N}}-\frac{\widehat{\mu}_v}
{M},~~ \widetilde{v}_{\alpha}+
\frac{t_{1-\beta/2,N-1}\widehat{\sigma}_{v}}{\sqrt{N}}-\frac{\widehat{\mu}_v}
{M}\right]
\end{eqnarray}
and
\begin{eqnarray}\label{eq.3.26}
\left[
\widetilde{c}_{\alpha}+
\frac{t_{\beta/2,N-1}\widehat{\sigma}_{c}}{\sqrt{N}}-\frac{\widehat{\mu}_c}{M},~~
\widetilde{c}_{\alpha}+
\frac{t_{1-\beta/2,N-1}\widehat{\sigma}_{c}}{\sqrt{N}}-\frac{\widehat{\mu}_c}
{M}\right],
\end{eqnarray}
where $\widehat{\sigma}_{v}$, $\widehat{\mu}_{v}$, $\widehat{\sigma}_{c}$ and $\widehat{\mu}_{c}$ are sample estimates of 
$\sigma_{v}$, ${\mu}_{v}$, ${\sigma}_{c}$ and ${\mu}_{c},$ respectively, and $t_{\gamma,L}$ represents the $\gamma$-quantile of a t-distribution with degree of freedom $L$.  The term ${f(v_\alpha)}$ in $\sigma_{v}$  can be estimated using Gaussian kernel density estimation (\cite{steckley2006}), and $\sigma_{c}$ can be estimated directly via sample average. 
The estimation of  ${\mu}_{v}$ and ${\mu}_{c}$ is more tricky as they involve expectation and gradient terms that could not be estimated directly. We will leave the detailed discussion on their estimation to Section 4. 
\begin{remark}
It is worth mentioning that we could probably avoid estimating $f$ directly by using more ``data driving" schemes such as sectioning and bootstrapping. This is an interesting future direction but beyond the scope of this paper. At the same time, the method we choose is easy to use and works quite well in the numerical experiments (see Section 5).
\end{remark}

Note that the CI in \eqref{eq.3.25} or \eqref{eq.3.26} only depends on $N$, when $N=o(M^2)$. In this case the bias term due to stochastic uncertainty 
is of the order $O(\frac{1}{M})$, and thus it will be asymptotically insignificant compared with the $O(\frac{1}{\sqrt{N}})$ error term. We refer to the CIs in \eqref{eq.3.25} and \eqref{eq.3.26} as ``CIs under Strong Assumption''.\\

The following Theorem \ref{thm.3.5} shows that CIs under Strong Assumption are asymptotically valid which means they will achieve a coverage probability of $(1-\beta)$.

\begin{theorem}\label{thm.3.5}
\textbf{Asymptotic Validity of CIs}.

Under Assumption \ref{asm.3.2}, the CIs defined in (\ref{eq.3.25}) and (\ref{eq.3.26}) are asymptotically valid when $N=o(M^2)$ exists, i.e.,
\begin{equation*}
\lim_{N,M\rightarrow \infty} Pr\{lb^s_{v}\le v_{\alpha}\le ub^s_{v}\}= 1-\beta ~~ \mbox{and} ~~
\lim_{N,M\rightarrow \infty} P\{lb^s_{c}\le c_{\alpha}\le ub^s_{c}\}= 1-\beta,
\end{equation*}
where $lb^s_{v}$ and $ub^s_{v}$ denote the lower and upper boundaries of the CI in (\ref{eq.3.25}), and $lb^s_{c}$ and $ub^s_{c}$ denote the lower and upper boundaries of the CI in (\ref{eq.3.26}), respectively.

\end{theorem}
\begin{proof}
See Appendix \ref{ap.e}.
\end{proof}

\begin{remark}
Here we only consider the case 	when $N=o(M^2)$ (or equivalently $K=0$) instead of $K\geq0$. This is because when $K>0$, it is difficult to guarantee that the estimate of the bias term is sufficiently accurate when $N$ and $M$ go to infinity, since we use a cubic basis function in the regression. 
\end{remark}

\color{black}
\section{Budget Allocation}\label{Sec4:Budget}

In practical simulation, usually there is a simulation budget that affects the choices of $N$ and $M$.
Intuitively, the outer sample size $N$ determines the simulation error due to input uncertainty, while the inner sample size $M$ determines the simulation error due to stochastic uncertainty.
Therefore, choosing $N$ and $M$ appropriately is critical to balance the trade-off between capturing input uncertainty and capturing stochastic uncertainty, and improve overall efficiency.

As shown in previous section, under Strong Assumption, the error of nested risk estimator $\widetilde{v}_\alpha$ (or $\widetilde{c}_\alpha$) could be decomposed into an error component caused by input uncertainty and a bias component caused by stochastic uncertainty.
Within this framework, \cite{gordy2010nested} proposes to minimize the asymptotic MSE, i.e., the summation of variance and squared bias, of $\widetilde{v}_\alpha$. The result is an (asymptotically) optimal budget allocation scheme, $N=O(M^2)$, that balances between the outer-layer sampling error and the inner-layer sampling bias.

\color{black}
An alternative approach to improving simulation efficiency is to formulate the optimal budget allocation problem by minimizing the width of the CI, noting that the CI in \eqref{eq.3.25} or \eqref{eq.3.26} are not centered at $\widetilde{v}_{\alpha}$ (or $\widetilde{c}_{\alpha}$). It may not even include $\widetilde{v}_{\alpha}$ (or $\widetilde{c}_{\alpha}$) when the bias dominates the standard deviation. So instead of minimizing the half width, one could minimize the largest possible difference between the true value ${v}_\alpha$ (or ${c}_\alpha$) and the point estimate $\widetilde{v}_\alpha$ (or $\widetilde {c}_\alpha$) under the specified high probability (such as $95\%$). This difference turns out to be the wider half of the CI in \eqref{eq.3.25} or \eqref{eq.3.26} and can be written as:
\begin{equation}\label{eq.4.1}
W_{v}(N,M)\overset{\triangle}= \frac{t_{1-\beta/2,N-1}\widehat{\sigma}_{v}}{\sqrt{N}}+|\frac{\widehat{\mu}_v}
{M}|,
\end{equation}
and
\begin{equation}\label{eq.4.2}
W_{c}(N,M)\overset{\triangle}= \frac{t_{1-\beta/2,N-1}\widehat{\sigma}_{c}}{\sqrt{N}}+|\frac{\widehat{\mu}_c}{M}|.
\end{equation}
 For simplicity, we will refer to $W_{v}$ (or $W_{c}$) as the ``wider-half CI width".

\color{black}
The budget allocation problem can be formulated as follows.  Let $C(N,M):=c_1 N+c_2 NM$ be the total computational cost, where $c_1$ is the cost for simulating one input parameter scenario, and $c_2$ is the cost for simulating one response sample.
Of course, there could be other minimization criteria such as the overall computational complexity, and they can be minimized in a similar manner. Let $CB$ be the total simulation budget.
Consider the following CI  width minimization problem

\begin{equation}\label{eq.4.3}
\begin{aligned}
&\min\limits_{N, M} && W_{v}(N,M) && \mbox{or}\;\; && \min\limits_{N, M} && W_{c}(N,M) \\
&s.t. && C(N,M)\le CB  &&  && s.t. && C(N,M)\le CB\\
&      &&  N\ge \Gamma_0,\; M\ge \Gamma_0 && && && N\ge \Gamma_0,\;
M\ge \Gamma_0,\; (1-\alpha)NM\ge \Gamma_0 \\
&      && N, M \in \mathbb{Z}^{+} && && && N, M \in \mathbb{Z}^{+}
\end{aligned}.
\end{equation}
Here the constraints $N\ge \Gamma_0$, $M\ge\Gamma_0$ and $(1-\alpha) NM\ge \Gamma_0$ are imposed to ensure the validity of a $t$-statistics, and a typical choice for $\Gamma_0$ is 30.

Before solving problem (\ref{eq.4.3}), we still need to compute or estimate the ``variance terms'' and ``bias terms'' $\sigma_{v}$, $\mu_{v}$, $\sigma_{c}$, and $\mu_{c}$ in the objective function, since in practice they are usually unknown or unavailable. A common fix is to run a pilot experiment with a small fraction of total simulation budget, and estimate the variance terms using the samples from the pilot experiment.
Let us use $\widetilde{\sigma}_{v}$, $\widetilde{\mu}_{v}$, $\widetilde{\sigma}_{c}$ and $\widetilde{\mu}_{c}$ to denote the estimates of $\sigma_{v}$, $\mu_{v}$, $\sigma_{c}$ and $\mu_{c}$ from the pilot experiment, respectively.
$\widetilde{\sigma}_{v}$ and $\widetilde{\sigma}_{c}$ could be the natural sample average estimates; however, they might be very inaccurate since it involves rare-event simulation with few samples.
For example, recall that
\begin{equation*}
\sigma_{c}^2=\frac{Var\left[\left(H(\theta)
-v_{\alpha}\right)^{+}\right]}{(1-\alpha)^2}
=\frac{1}{(1-\alpha)^2}\left\{\mathbb{E}\left[\left(\left(H(\theta)-
v_{\alpha}\right)^+\right)^2\right]-
\left(\mathbb{E}\left[\left(H(\theta)
-v_{\alpha}\right)^{+}\right]\right)^2\right\}.
\end{equation*}
This indicates that estimation of $\sigma_{c}^2$ is at least as difficult as estimation of $v_\alpha$. Using naive sample average to estimate $\sigma_{c}$ causes most of the samples to be ineffective, and thus resulting in an inaccurate estimate $\widetilde{\sigma}_{c}$.
In fact, theoretically only $(1-\alpha)$ fraction of the samples will be effective; since $\alpha$ is close to 1, the percentage of effective samples is small.
To be more specific, suppose $\alpha=0.99$ and $N=100$ scenarios of $H(\theta)$ are generated in the pilot experiment.
Then theoretically only one scenario will be effective and used in the estimation, since the rest $99$ scenarios result in a simple value of $0$.

The issue with naive sample average method is that the information about the underlying distribution carried by the ineffective samples is not utilized.
In contrast, a good estimation method usually makes use of the information carried by all the samples.
For example, using (adaptive) importance sampling turns some of the ineffective samples into effective samples,
and thus improves accuracy; however, this approach is not readily applicable here because we lack the knowledge about the p.d.f. of the mean response distribution.

Next, we will propose a new approach to estimating the variance terms that exploits the information carried by all the samples generated in the pilot experiment. Recall that
$$
\sigma^2_{v}=\alpha(1-\alpha)/f^2(v_\alpha), \quad \mu_v=\frac{-\Lambda^\prime(v_{\alpha})}
{f(v_{\alpha})},
$$
and
$$
\sigma^2_{c}=Var\left[\left(H(\theta)
-v_{\alpha}\right)^{+}\right]/(1-\alpha)^2,\quad
\mu_c=\frac{\Lambda(v_{\alpha})}
{(1-\alpha)},
$$
where $$\Lambda(t)=1/2 f(t)\mathbb{E}[\tau_\theta^2|H(\theta)=t].$$The challenges are three-fold: (i) the lack of an explicit formula for $f(\cdot)$; (ii) the lack of a functional representation for $\tau^2(\cdot)$ in the function $\Lambda(t)$, where $\tau^2(y):=\mathbb{E}[\tau^2_\theta|H(\theta)=y]$; (iii) the lack of the gradient of $\Lambda(t)$.

To address the first challenge, we apply a technique called ``density projection''.
That is, we project the discrete empirical distribution of $H(\theta)$ onto a parameterized family of continuous densities.
Then the resultant projection, which is a continuous density, will be used as an approximation of $f(\cdot)$, and $\widetilde{\sigma}_{v}$ and $\widetilde{\sigma}_{c}$ are computed via numerical integration.
The detailed description of density projection is as follows.

A \emph{projection mapping} from a space of probability distributions $\mathcal{P}$ to another space consisting of a parameterized family of densities $\mathcal{F}$, denoted as
$Proj_\mathcal{F}: \mathcal{P} \rightarrow \mathcal{F}$, is defined by \begin{equation}\label{eq.4.5}
Proj_{\mathcal{F}}(g)\overset{\triangle}=\arg \min\limits_{f\in \mathcal{F}} D_{KL}(g\parallel f), \quad \forall g\in \mathcal{P},
\end{equation}
where $D_{KL}(g\parallel f)$ denotes the \emph{Kullback-Leibler (KL) divergence} between $g$ and $f$, which is
\begin{equation*}
D_{KL}(g\parallel f)\overset{\triangle}=\int g(x)\log \frac{g(x)}{f(x)}dx.
\end{equation*}
Here note that the densities $g$ and $f$ are assumed to have the same support.
Hence, the projection of $g$ on $\mathcal{F}$ has the minimum KL divergence from $g$ among all densities in $\mathcal{F}$.
Loosely speaking, the projection of $g$ on $\mathcal{F}$ is the best approximation of $g$ one can find in $\mathcal{F}$.
When $\mathcal{F}$ is an exponential family of densities, which includes common families of densities such as Gaussian, the minimization problem (\ref{eq.4.5}) has an analytical solution.
Note that this technique utilizes the information carried by all the samples.

\begin{remark}\label{rem.4.1}
If i.i.d. samples of $g$ are generated to compute  $Proj_{\mathcal{F}}(g)$, then the proposed density projection technique is equivalent to maximum likelihood estimation. Furthermore, if $\mathcal{F}$ is an exponential family of densities with sufficient statistics that consist of polynomials, then density projection is equivalent to method of moments.
\end{remark}

To address the second challenge, we apply regression for $\tau^2(y)$ onto the space of $H(\theta)$, and use the samples from the pilot experiment to train the regression model.
Simple numerical tests show that a polynomial regression with basis functions consisting of polynomials (degree$\le 3$) of $H(\theta)$ is sufficiently good.

\color{black}
The third challenge is resolved naturally to this end because we have the closed form of $\Lambda(t)$ as a function of $t$. In particular, $f(t)$ is now a normal probability density function, and $\tau^2(y)$ is a polynomial function. Thus, we can compute the gradient analytically.
\color{black}

After plugging the approximate  terms $\widetilde{\sigma}_{v}$, $\widetilde{\mu}_{v}$, $\widetilde{\sigma}_{c}$ and $\widetilde{\mu}_{c}$ into problem (\ref{eq.4.3}), it remains to solve the minimization problem.
Solving it analytically to optimality is unlikely because the problem might not possess structural properties such as convexity. \textcolor{black}{In particular, the first constraint $C(N,M)\leq CB$ is concave.}
Alternatively, we can enumerate a reasonable amount of candidate allocation schemes (e.g., a two-dimensional grid of feasible allocation schemes),
and choose a scheme that yields the smallest CI width.

We also point out that it is beneficial to consider a more sophisticated budget allocation scheme in which the inner sample size varies across different input (parameter) scenarios.
For example, in the estimation of $v_\alpha$, the input scenarios that heavily affect estimation accuracy are the ones with mean responses close to $v_\alpha$.
In particular, for a specific input scenario, it affects estimation accuracy if the true mean response of that input scenario falls into one side of $v_\alpha$ while its estimation falls into the other side.
In this case, the inner sample size for this input scenario should be increased to reduce the probability of such event.
This problem has been studied in the setting of nested credit risk assessment using ranking and selection (\cite{broadie2011efficient}) and screening (\cite{lan2010confidence}), etc.

\section{Numerical Experiments}\label{Sec5:Numerical}

\subsection{CIs under Strong Assumption}\label{Sec5.subsec1:Comparison}

We first use a simple numerical example from \cite{gordy2010nested} to show the validity of our CI procedures under Strong Assumption.
In particular, consider $H(\theta;\xi)=\mathcal{N}(0,1)+\mathcal{N}(0,1)$, a summation of two independent standard normal random variables.
In \cite{gordy2010nested}, the first $\mathcal{N}(0,1)$ represents the (outer-layer) portfolio loss distribution and the second $\mathcal{N}(0,1)$ represents the (inner-layer) pricing error. Clearly, this example does not fit into our input uncertainty framework.
The reason for using it is that the exact risk values, and all variance and bias parameters admit closed-form expressions.
Thus, strong CI procedure are precise.\color{black}

Performance measures of interest include wider-half CI widths and actual coverage probability, i.e., the probability that the true risk value falls into the simulated CI.
In particular, we will run the simulation 1000 times independently and identically to compute the two performance measures, in which the budget allocation scheme from minimizing the wider-half CI widths in previous section is employed.
The results for VaR and CVaR are summarized in Table \ref{table.5.1}.

\begin{table}[htb]
\color{black}
\tbl{Strong CI Procedures in VaR and CVaR Estimation.\label{table.5.1}}
{\begin{tabular}{|c|cccc|cccc|}
\hline
 & \multicolumn{4}{c|}{\textbf{VaR}} & \multicolumn{4}{c|}{\textbf{CVaR }}\\
\hline
$C(N, M)$     &$N_w$  &$M_w$     &Wider Half  &Coverage  &$N_s$   &$M_s$  &Wider Half    &Coverage     \\
&        &       & CI  Width   &Probability   &       &       &  CI Width&Probability  \\
\hline
$10^4$ &$865$  &$12$   &$0.2096$ &$93.9\%$   &$824$    &$13$  &$0.2477$    &$94.1\%$  \\
\hline
$10^5$ &$4015$  &$25$  &$0.0983$ &$94.5\%$  &$3826$   &$27$  &$0.1163$    &$94.4\%$   \\
\hline
$10^6$ &$18634$ &$54$  &$0.0456$ &$95.1\%$  &$17758$  &$57$  &$0.0544$    &$95.1\%$   \\
\hline
$10^7$ &$86491$ &$116$ &$0.0212$ &$95.5\%$  &$82429$  &$122$  &$0.02533$   &$95.2\%$   \\
\hline
\end{tabular}}
\begin{tabnote}
\Note{Note:}{The risk level of interest $\alpha=0.95$, the target confidence level $(1-\beta)=0.95$, and the total simulation cost $C(N, M)= NM+N$. The pair $(N_w, M_w)$ is the budget allocation obtained by minimizing the wider-half CI width. The coverage probabilities are obtained via $1000$ independent and identical runs of simulation. }
\end{tabnote}%
\end{table}

The numerical results show that: 1) As expected, strong CI procedure generates CIs with coverage probabilities around 95\%. 
2) When we increase the total budget, the coverage will increase while the wider-half CI width will decrease. Thus, the strong procedure can provide us a very good estimation (high coverage, small CI) if we have large budget.
3) the optimal budget allocations for VaR and CVaR are almost the same, which means we can minimize the two wider-half CI widths at the same time.

\subsection{Sharing Economy Model}\label{Sec5.subsec2:MM1}

Let us consider another example for risk quantification under input uncertainty--- a sharing economy model. This new type of economy refers to businesses such as local delivery, car sharing and house sharing. We can model it as a two-sided market in which two distinct user groups, buyers and sellers, trade with each other under the regulation made by the organizer/agent. In general the organizer is responsible for pricing, with the aim to maximize the revenue or clear the marketplace. Usually, both demand and supply depend on the price. Specifically, when the price is higher, more sellers will come while more buyers will be lost, and vice versa. We can further model this as a queueing system (see \cite{banerjee2015pricing}) and make the following assumptions:
\begin{enumerate}
     \item Buyers and sellers' arrivals follow Poisson processes.
     \item Buyers and sellers are homogeneous in terms of price sensitivity. 
    \item Buyers and sellers all follow the First Come First Serve rule.
	\item Every buyer only needs one product, and at the same time each seller only has one product available to sell.
	\item Lost Order: If a buyer enters the market and finds no sellers left, the buyer will exit the market immediately.
\end{enumerate}  
Denote the basic arrival rates of buyers and sellers by $\Lambda_o$ and $M_o$, respectively, and $\theta_o :=(\Lambda_o,M_o)$.
Let $p$ be the price. In general, when a buyer (seller) enters the market, after learning the price, the buyer (seller) will buy/sell the product (if available) with probability $f(p)$ ($g(p)$ for sellers). Note that $f(p)$ and $g(p)$ could be viewed as the functions that describe users' sensitivities to the price. One of the common choices for $f(p)$ and $g(p)$ is
\begin{align}
f(p)&=\frac{ exp(-\alpha p)}{1+exp(-\alpha p)}, \\
g(p)&=\frac{1}{1+exp(-\beta p)}
\end{align}
where $\alpha$ and $\beta$ are the price sensitivity coefficients. Since in general buyers are more sensitive to the price change, we also assume $\alpha>\beta$. Notice that $f(p)$ is strictly decreasing, and $g(p)$ is a strictly increasing function. 
Thanks to the nice properties of general Poisson processes, the actual arrivals of buyers and sellers (meaning the ones that do offer to buy or sell) under price $p$ also follow Poisson processes with rates $\lambda(p)$ and $\mu(p)$, respectively, where
\begin{align}
	\lambda(p) &=\Lambda_of(p)\\
	\mu(p) &=M_o g(p).
\end{align}

The response of interest for the organizer is the unfulfilled rate $H(\theta_o,p)$ (i.e the probability of an order from buyers being lost) because it measures to what extent the marketplace is cleared while managing the service rate and the queue length of the sellers. In particular, the organizer could control the unfulfilled rate $H(\theta_o,p)$ by adjusting the price dynamically.

\begin{remark}\label{remark.5.1}
Usually, the organizer will not try to find a price that minimizes $H(\theta_o,p)$. Intuitively, when the price is very high, there will be more sellers than buyers in the market. In this case, the unfulfilled rate is close to 0. However, this scenario might not be of the best interest to the organizer since the organizer's objective might be maximizing the revenue or social welfare (which is very common in a two-sided market). Another example is that the organizer might want to maximize the expected number of fulfilled orders while controlling the unfulfilled rate. That is, we need to solve the following stochastic optimization problem:
\begin{align*}
	\max_p \quad &E[N(\theta_o,p)(1-H(\theta_o,p))]\\
	s.t. \quad & H(\theta_o,p)\leq0.05 ,\\
	&p>0,
\end{align*}
where $N(\theta_o,p)$ is the total number of buyers entering the market. This problem could be solved by simulation optimization methods, which is beyond the scope of this paper.
\end{remark}

Here we focus on estimating the unfulfilled rates under different prices. The challenge is that we do not know the basic arrival rates $\theta_o$ exactly. In practice, we have to collect data to estimate them, and hence input uncertainty plays a crucial role. Our objective is to estimate the risk associated with unfulfilled rate due to input uncertainty.

In this numerical experiment,  the values of $\Lambda_o$ and $M_o$ are known to us (the judges) but not known to the experimenter. Take $\Lambda_o=5$, $M_o=2$, $\alpha=0.2$, $\beta=0.1$ and $p=2,3,4,5$.
To model input uncertainty, we take a Bayesian approach to construct the belief distribution on input parameters---the basic Poisson arrival rates $\Lambda_o$ and $M_o$.
Specifically, assume non-informative priors for both $\Lambda_o$ and $M_o$, i.e., $p_o(\Lambda_o)\propto 1/\Lambda_o$ and $p_o(M_o)\propto 1/M_o$. Based on $n=10,100,10000$ historical observations of $\Lambda_o$ and $M_o$ (drawn from the corresponding distributions with the true parameters), a Bayesian updating is applied to obtain the posterior distributions of $\Lambda_o$ and $M_o$.
In particular, denote the historical observations of $\Lambda_o$ by $\mathbf{x}=(x_1,...,x_n)$. Then the updating on the posterior distribution of $\Lambda_o$ is carried out analytically and leads to  $p(\Lambda_o|\mathbf{x})=\Lambda_o^{n-1}\exp{(-\Lambda_o\sum_{i=1}^{n}x_i)}$, which is a Gamma distribution with shape parameter $n$ and scale parameter $1/(\sum_{i=1}^{n}x_i)$.
Similarly, let $\mathbf{y}=(y_1,...,y_n)$ be the historical observations of $M_o$.
Then the posterior distribution of $M_o$ is $p(M_o|\mathbf{y})=M_o^{n-1}\exp{(-M_o\sum_{i=1}^{n}y_i)}$---a Gamma distribution with shape parameter $n$ and scale parameter $1/(\sum_{i=1}^{n}y_i)$.
The objective is to estimate $v_\alpha$ and $c_\alpha$ ($\alpha=0.90, 0.95, 0.99$) of unfulfilled rate  w.r.t. the posterior parameter distributions $p(\Lambda_o|\mathbf{x})$ and $p(M_o|\mathbf{y})$, and construct the associated $100(1-\beta)\%$ CIs ($\beta=0.05$).

Before we apply the strong CI procedure, first we need verify whether the Strong Assumption holds. In this example, the parameter $\theta_o =(\Lambda_o,M_o)$. What we want to estimate is a probability, so here $h(\theta_o,\xi)$ is a Bernoulli random variable with the probability of success equals to $H(\theta_o,p).$ Thus it has a finite conditional second moment. In addition, $\bar{\mathcal{E}}_M$ can be written as $B(M,h))/\sqrt{M}-\sqrt{M}h,$ where $B(\cdot,\cdot)$ is a Binomial random variable. Notice that in practice we usually choose $M$ large enough to make sure the estimate is accurate. Hence, we can use the normal approximation $\mathcal{N}(Mh,Mh(1-h))$ here to replace the Binomial distribution.  We denote this conditional probability density function as $f_{M,h}(e)$ whose second derivative exists and is continuous.  Then the joint density $p_M(h,e)$ is the product of $f_{M,h}(e)$ and $f(h).$ The exponential term in the normal distribution will ensure the second and third items in the Strong Assumption are satisfied.

In particular, we draw $N=10000$ input parameter scenarios from $p(\Lambda_o|\mathbf{x})$ and $p(M_o|\mathbf{y})$.
Furthermore, for each input parameter scenario, we draw $M=2000$ samples of buyer's arrival and count how many of their orders are lost.
Finally, $v_\alpha$ and $c_\alpha$ of the unfulfilled rate are estimated via (\ref{eq.3.5}) and (\ref{eq.3.6}), respectively.
The simulation results are summarized in Table \ref{table.5.2} and \ref{table.5.3}.
\begin{table}[htb]
\color{black}
\tbl{VaR (with 95$\%$ CI) of Unfulfilled Rate.\label{table.5.2}}
{\begin{tabular}{ccllll}
\hline
$p$  &  $n$  & Mean        & $VaR_{\alpha_1}$     &$VaR_{\alpha_2}$  &$VaR_{\alpha_3}$ \\
     &            &Half Width     &wider-half CI width
   &wider-half CI width    &wider-half CI width  \\
\hline\hline
2 & 10    & 0.2109   & 0.5275   & 0.601    & 0.722    \\
  &       & 4.20$\times10^{-3}$ & 7.50$\times10^{-3}$ & 8.20$\times10^{-3}$ & 1.12$\times10^{-2}$ \\ \hline
2 & 100   & 0.4875   & 0.58     & 0.602    & 0.6425   \\
  &       & 1.50$\times10^{-3}$ & 3.10$\times10^{-3}$ & 3.70$\times10^{-3}$ & 5.80$\times10^{-3}$ \\ \hline
2 & 10000 & 0.4598   & 0.4875   & 0.4955   & 0.509    \\
  &       & 4.70$\times10^{-4}$ & 4.40$\times10^{-3}$ & 5.60$\times10^{-3}$ & 8.30$\times10^{-3}$ \\ \hline
3 & 10    & 0.6743   & 0.836    & 0.862    & 0.8975   \\
  &       & 3.00$\times10^{-3}$ & 3.00$\times10^{-3}$ & 3.10$\times10^{-3}$ & 2.90$\times10^{-3}$ \\ \hline
3 & 100   & 0.2207   & 0.359    & 0.3905   & 0.45     \\
  &       & 2.10$\times10^{-3}$ & 3.80$\times10^{-3}$ & 4.40$\times10^{-3}$ & 6.80$\times10^{-3}$ \\ \hline
3 & 10000 & 0.3508   & 0.3825   & 0.3915   & 0.408    \\
  &       & 5.30$\times10^{-4}$ & 3.80$\times10^{-3}$ & 4.90$\times10^{-3}$ & 7.30$\times10^{-3}$ \\ \hline
4 & 10    & 0.1201   & 0.401    & 0.4995   & 0.638    \\
  &       & 3.40$\times10^{-3}$ & 8.10$\times10^{-3}$ & 7.90$\times10^{-3}$ & 1.44$\times10^{-2}$ \\ \hline
4 & 100   & 0.3674   & 0.4805   & 0.508    & 0.558    \\
  &       & 1.80$\times10^{-3}$ & 3.50$\times10^{-3}$ & 4.20$\times10^{-3}$ & 5.90$\times10^{-3}$ \\ \hline
4 & 10000 & 0.235    & 0.2705   & 0.2805   & 0.3      \\
  &       & 5.80$\times10^{-4}$ & 3.00$\times10^{-3}$ & 4.00$\times10^{-3}$ & 6.20$\times10^{-3}$ \\ \hline
5 & 10    & 0.1371   & 0.4305   & 0.5175   & 0.65     \\
  &       & 3.60$\times10^{-3}$ & 7.30$\times10^{-3}$ & 9.20$\times10^{-3}$ & 1.31$\times10^{-2}$ \\ \hline
5 & 100   & 0.1128   & 0.252    & 0.2905   & 0.36     \\
  &       & 1.90$\times10^{-3}$ & 4.20$\times10^{-3}$ & 5.00$\times10^{-3}$ & 8.10$\times10^{-3}$ \\ \hline
5 & 10000 & 0.008    & 0.127    & 0.139    & 0.16     \\
  &       & 6.10$\times10^{-4}$ & 2.00$\times10^{-3}$ & 2.80$\times10^{-3}$ & 4.60$\times10^{-3}$\\
\hline

\end{tabular}}
\begin{tabnote}
\Note{Note:}{The experiment parameters are: $\Lambda_o=5$, $M_o=2$, $N=10000$, $M=2000$, $\alpha_1=0.90$, $\alpha_2=0.95$, and $\alpha_3=0.99$.}
\end{tabnote}%
\end{table}

\begin{table}[htb]
\color{black}
\tbl{CVaR (with 95$\%$ CI) of Unfulfilled Rate.\label{table.5.3}}
{\begin{tabular}{ccllll}
\hline
$p$  &  $n$  & Mean       & $CVaR_{\alpha_1}$     &$CVaR_{\alpha_2}$  &$CVaR_{\alpha_3}$ \\
     &            &Half CI Width     &wider-half CI width
   &wider-half CI width    &wider-half CI width$\times10^{-3}$  \\
\hline\hline
	2 & 10    & 0.2109   & 0.6181   & 0.6722   & 0.7611   \\
  &       & 4.20$\times10^{-3}$ & 7.30$\times10^{-3}$ & 8.40$\times10^{-3}$ & 1.06$\times10^{-2}$ \\ \hline
2 & 100   & 0.4875   & 0.6078   & 0.6253   & 0.6585   \\
  &       & 1.50$\times10^{-3}$ & 3.60$\times10^{-3}$ & 4.30$\times10^{-3}$ & 6.30$\times10^{-3}$ \\ \hline
2 & 10000 & 0.4598   & 0.4975   & 0.5039   & 0.517    \\
  &       & 4.70$\times10^{-4}$ & 5.90$\times10^{-3}$ & 7.00$\times10^{-3}$ & 9.80$\times10^{-3}$ \\ \hline
3 & 10    & 0.6743   & 0.8665   & 0.8845   & 0.9128   \\
  &       & 3.00$\times10^{-3}$ & 2.60$\times10^{-3}$ & 2.70$\times10^{-3}$ & 3.80$\times10^{-3}$ \\ \hline
3 & 100   & 0.2207   & 0.3999   & 0.4265   & 0.4747   \\
  &       & 2.10$\times10^{-3}$ & 4.20$\times10^{-3}$ & 5.20$\times10^{-3}$ & 8.10$\times10^{-3}$ \\ \hline
3 & 10000 & 0.3508   & 0.394    & 0.4014   & 0.4159   \\
  &       & 5.30$\times10^{-4}$ & 5.10$\times10^{-3}$ & 6.10$\times10^{-3}$ & 8.60$\times10^{-3}$ \\ \hline
4 & 10    & 0.1201   & 0.5167   & 0.583    & 0.6873   \\
  &       & 3.40$\times10^{-3}$ & 9.10$\times10^{-3}$ & 9.90$\times10^{-3}$ & 1.32$\times10^{-2}$ \\ \hline
4 & 100   & 0.3674   & 0.5163   & 0.5394   & 0.5795   \\
  &       & 1.80$\times10^{-3}$ & 4.00$\times10^{-3}$ & 4.90$\times10^{-3}$ & 7.50$\times10^{-3}$ \\ \hline
4 & 10000 & 0.235    & 0.2834   & 0.2918   & 0.3076   \\
  &       & 5.80$\times10^{-4}$ & 4.10$\times10^{-3}$ & 5.00$\times10^{-3}$ & 7.20$\times10^{-3}$ \\ \hline
5 & 10    & 0.1371   & 0.5336   & 0.5982   & 0.6991   \\
  &       & 3.60$\times10^{-3}$ & 8.40$\times10^{-3}$ & 9.60$\times10^{-3}$ & 1.35$\times10^{-2}$ \\ \hline
5 & 100   & 0.1128   & 0.3012   & 0.3328   & 0.3927   \\
  &       & 1.90$\times10^{-3}$ & 4.70$\times10^{-3}$ & 5.90$\times10^{-3}$ & 9.80$\times10^{-3}$ \\ \hline
5 & 10000 & 0.008    & 0.1419   & 0.1515   & 0.1692   \\
  &       & 6.10$\times10^{-4}$ & 2.80$\times10^{-3}$ & 3.50$\times10^{-3}$ & 5.50$\times10^{-3}$\\
  \hline
 \end{tabular}}
 \begin{tabnote}
\Note{Note:}{The experiment parameters are: $\Lambda_o=5$, $M_o=2$, $N=10000$, $M=2000$, $\alpha_1=0.90$, $\alpha_2=0.95$, and $\alpha_3=0.99$.}
\end{tabnote}%
\end{table}
We have the following observations:
\begin{itemize}
\item[(1)] In general, there is a significant gap between the mean (column 3)  and VaR or CVaR  (columns 4 to 6) of unfulfilled rate w.r.t. input uncertainty. It implies that risk quantification in stochastic simulation under input uncertainty is necessary. Moreover, when $n$ is really small (such as $n=10$), there is no clear pattern how the mean value changes when $p$ increases. This is because the error caused by input uncertainty is too large and overwhelms the estimation.

\item[(2)] When $n$ is large, the gap between the mean and VaR or CVaR becomes small.
Intuitively, as more input data become available, the belief distribution on input parameter becomes more concentrated on the values close to the true one. Therefore, loosely speaking, the mean response distribution is also more concentrated on the values close to the true mean response, and essentially reduces the risk of large unfulfilled rate.

\item[(3)] Under the same level of input uncertainty, especially when $n$ is small, we can see the gap between the mean and VaR or CVaR becomes more significant as $p$ increases. For example, when $n=100$, $VaR_{\alpha_1}$ is  only 1.2 times of mean for $p=2$ while this number is more than 3 for $p=5$. This is because when $\mu(p)$ approaches the buyers' arrival rate $\lambda(p)$, the system becomes less stable and the risk in simulation due to input uncertainty is more significant. Therefore, more input data is required to reduce such risk to an acceptable level.
\end{itemize}

In order to show how input uncertainty might affect the pricing scheme, we further study how VaR and CVaR estimates behave around the optimal price (about 5 dollars) under different levels of input uncertainty. In particular, we take three different input data sizes $n=100,\;1000,\;10000$.  For each $n$, we estimate the mean,  $\textit{VaR}_{0.95}$ and  $\textit{CVaR}_{0.95}$ and their $95\%$ CIs under the price range of  $4.8$ to $5.1$. The results are shown in Figure \ref{fig.1}, where lines with different colors show the trends of the mean, VaR and CVaR, and bars at different points represent the corresponding CIs. 
\begin{figure}[htb]
\centering
\begin{tabular}{c}
	\hspace{-5mm}
    \includegraphics[width=0.60\textwidth]{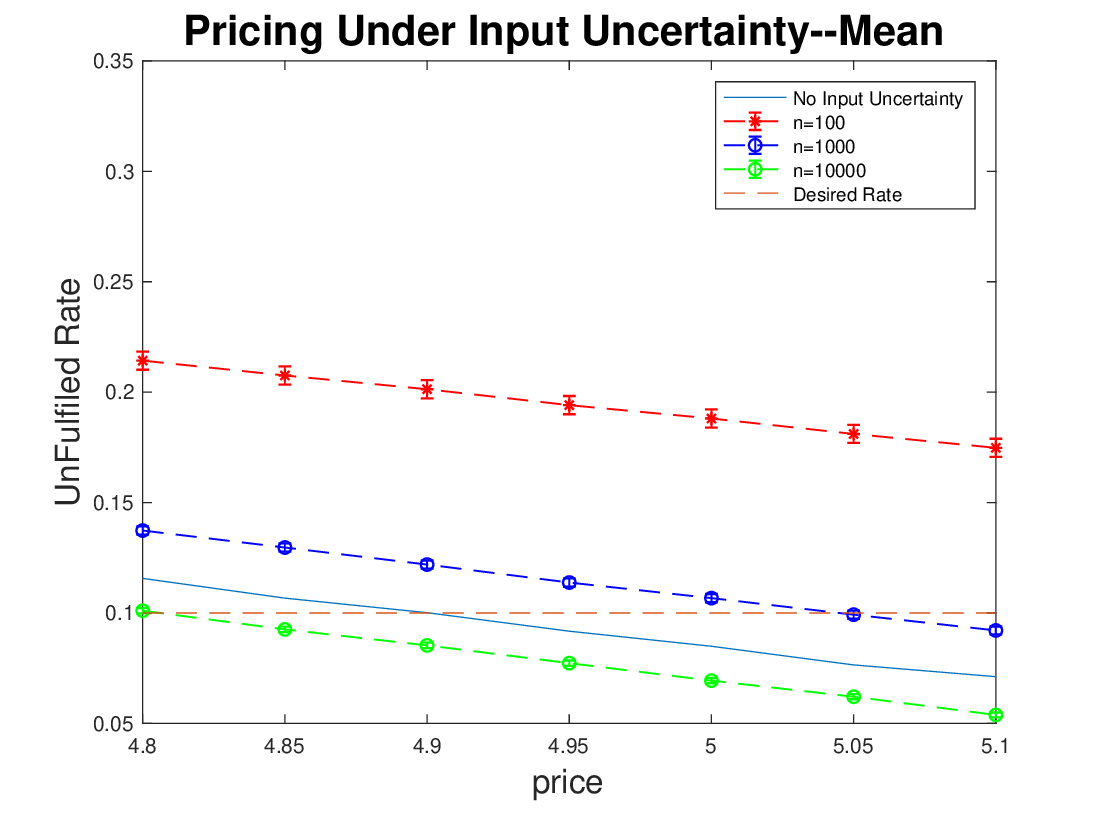}\\
\end{tabular}
\begin{tabular}{cc}
	\hspace{-5mm}
	\includegraphics[width=0.50\textwidth]{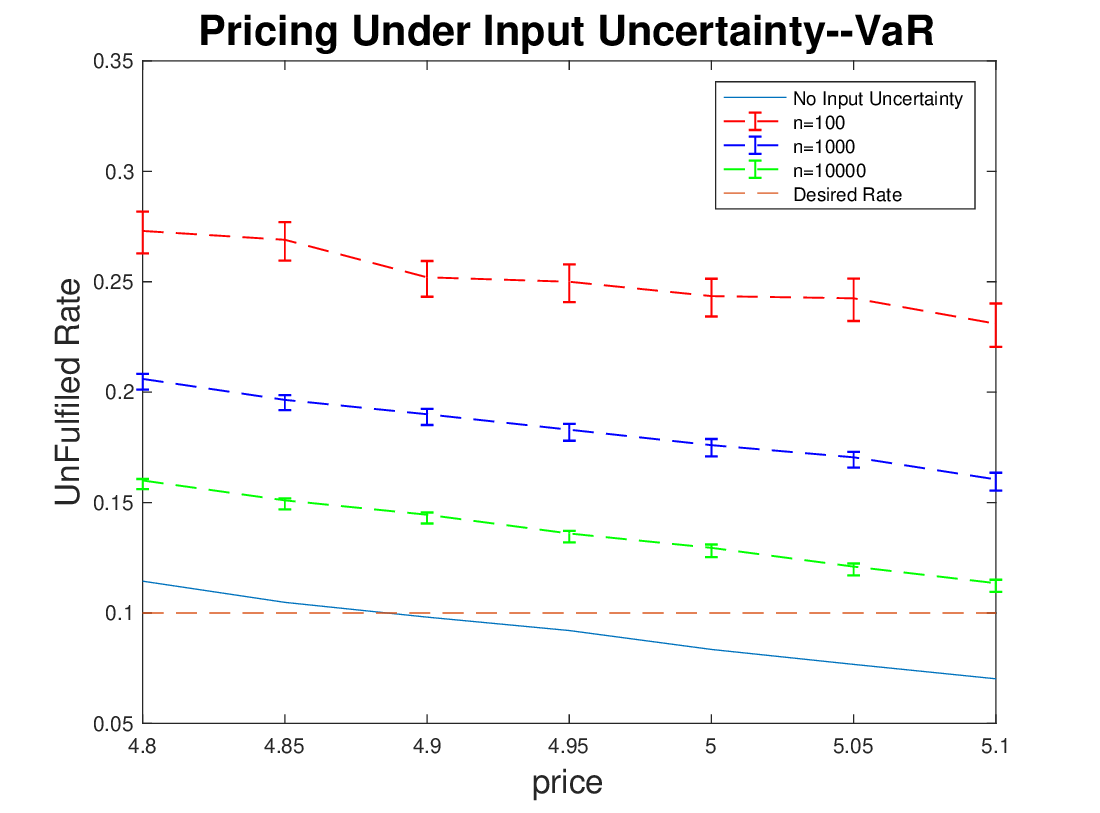} 
	\hspace{-5mm}
	\includegraphics[width=0.50\linewidth]{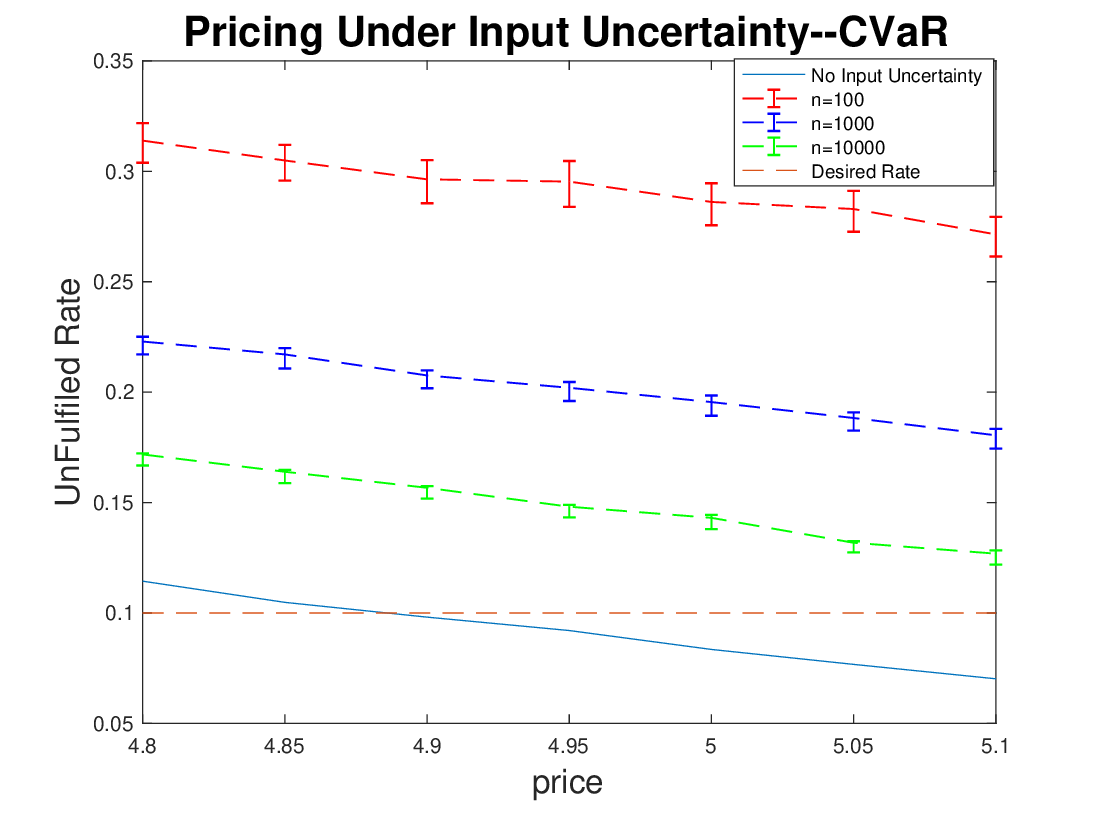}\\
\end{tabular}
\centering
\caption{ Behavior of VaR, CVaR, Mean around Optimal Price}\label{fig.1}
{\footnotesize \indent The experiment parameters are: $\Lambda_o=5$, $M_o=2$,$\alpha=0.95$, $N=2500$,$M=2000$.}
\end{figure}

We have the following observations:
\begin{itemize}
\item[(1)]In all three plots, as $n$ increases, mean, VaR, and CVaR approach the real unfulfilled rate (solid blue line).

\item[(2)] Obviously, all the plots of VaR and CVaR are above the true unfulfilled rate, indicating the risk caused by input uncertainty exists and cannot be ignored. If input uncertainty is not considered carefully, the price we find will be much different from the real optimal price. In this case,  the organizer will lose their profits because of the low price or the lack of orders.

\item[(3)] Under the same price, input uncertainty will greatly affect the CIs' widths for VaR and CVaR. In particular, when we use more observations to estimate $\Lambda_o$ and $M_o$, the CIs will be narrower. Together with the first observation, we can minimize the influence of the input uncertainty by collecting more input data.

\item[(4)] It is not clear from this figure how does price affect CIs' width due to the small price range of 4.8 to 5.1. But from previous results, we can know that the length (half-width for the mean, wider-half CI width for VaR and CVaR) increases as the price increases in certain cases. 
\end{itemize}

We finally study the associated budget allocation problem. Note that for VaR estimation and CVaR estimation, the budget allocation problem might yield different optimal allocation schemes.
Let $C(N, M)=NM+N$ and $CB=5\times 10^6$. We use $N_{pilot}=100$ outer scenarios and $M_{pilot}=50$ inner samples for each scenario in the pilot experiment to guide the budget allocation in the actual experiment.
In total, only $0.1\%$ percent of total budget is consumed, so the budget for the actual experiment is minimally affected.
To show the effectiveness of the pilot experiment, we plot the wider-half CI widths for different choices of $N$ in Figure \ref{fig.2}, where the blue curves are the wider-half CI widths calculated using terms estimated from the pilot experiment, and the red curves are the wider-half CI widths calculated using the true values obtained by brute-force simulation (i.e., using extremely large sample sizes).

\begin{figure}[htb]
\centering
\begin{tabular}{cc}
     \hspace{-7mm}
    \includegraphics[width=0.55\textwidth]{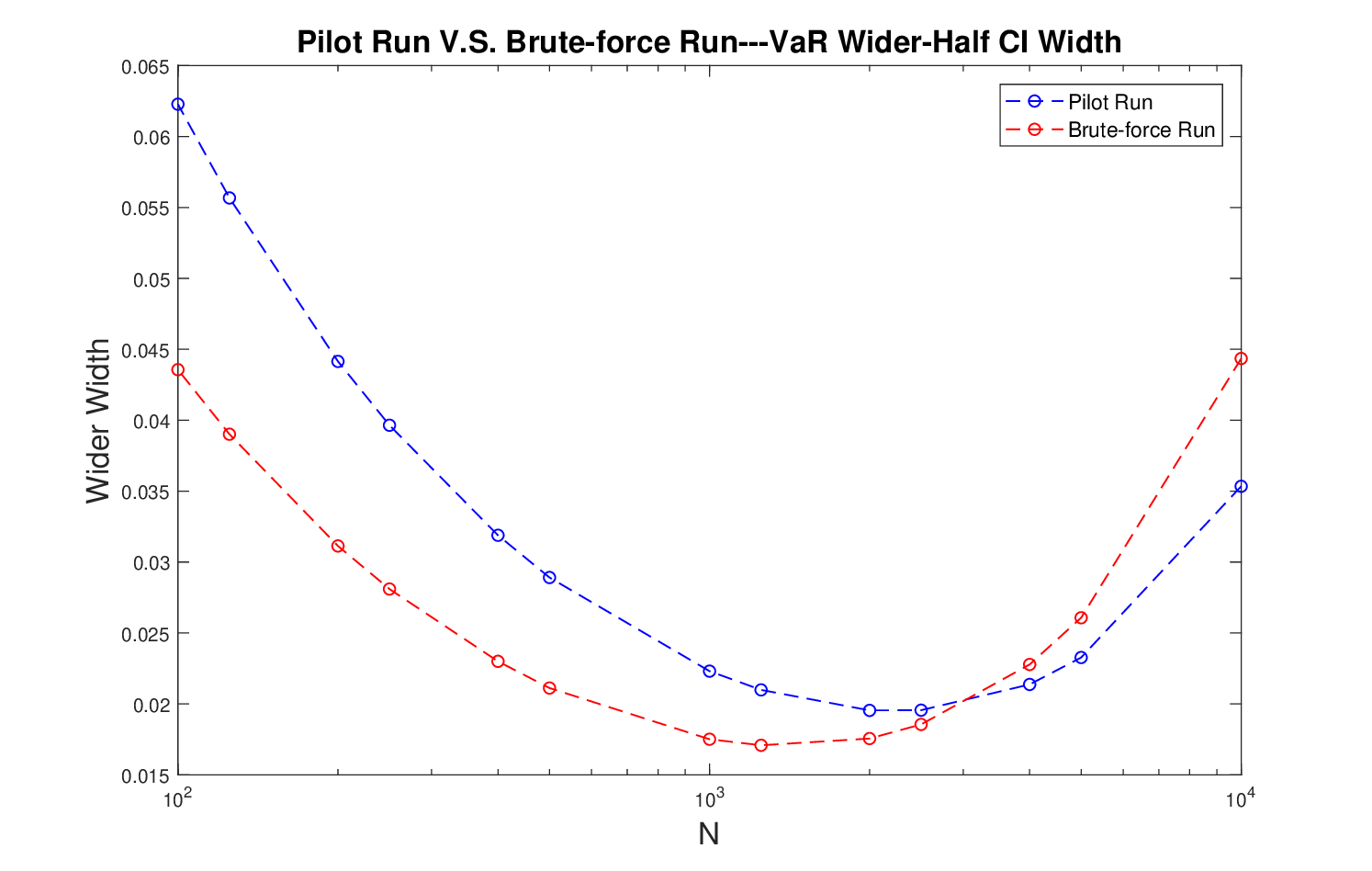} &
    \hspace{-12mm}
    \includegraphics[width=0.55\linewidth]{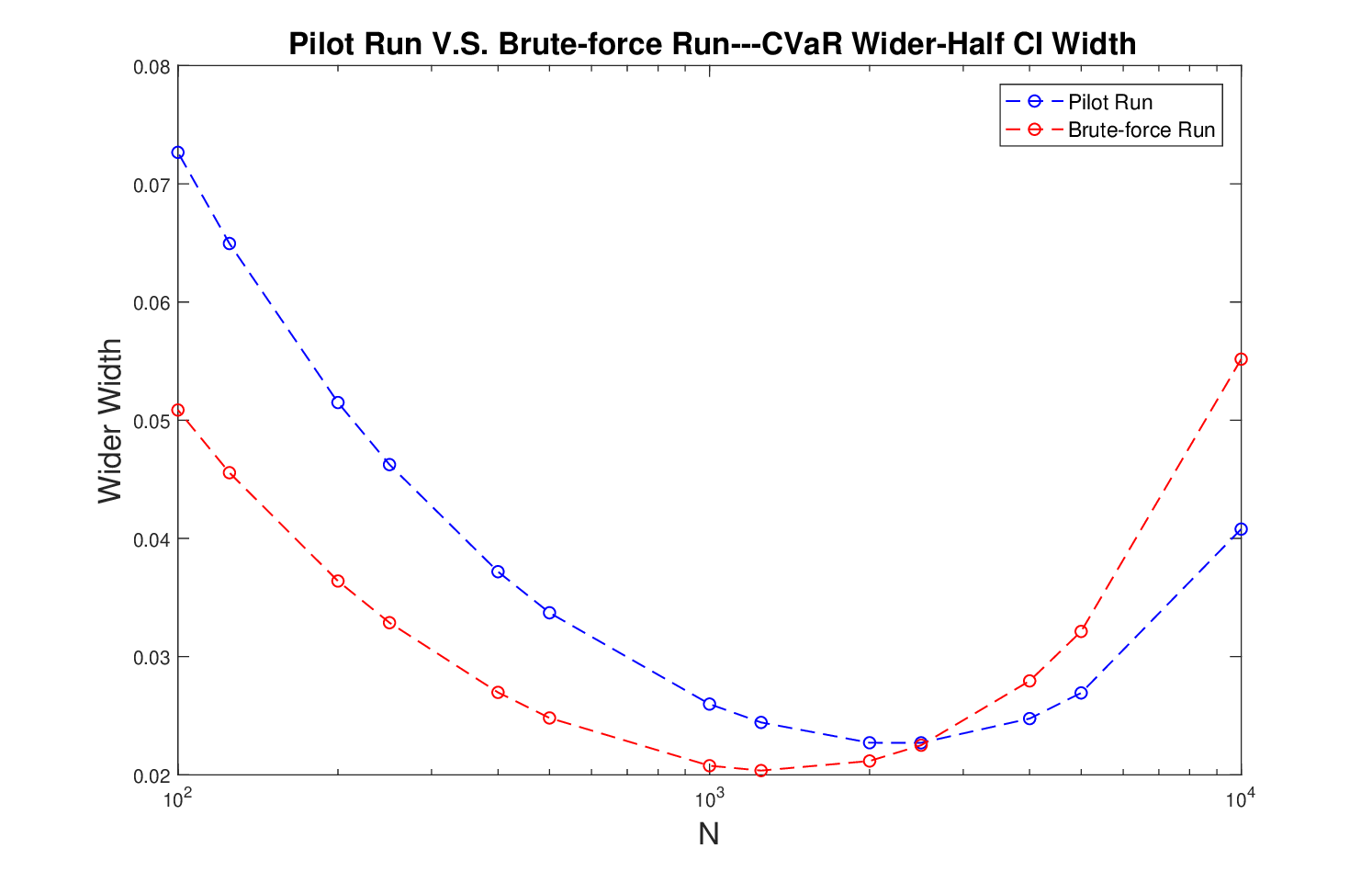}\\
\end{tabular}
\centering
\caption{ VaR and CVaR wider-half CI width: Pilot Run V.S. Brute-Force Run}\label{fig.2}
{\footnotesize \indent The experiment parameters are: $\Lambda_o=5$, $M_o=2$, $p=4$, $\alpha=0.95$, and the size of input data $n=100$.}
\end{figure}

We have the following observations:

\begin{itemize}
\item[(1)] In both plots, although there is a non-negligible gap between the wider width (blue curve) computed using the  terms estimated from the pilot experiment and the true wider width (red curve), the curves follow the same trend and their minima almost coincide. 
This implies that solving the formulated budget allocation problem could identify the optimal budget allocation scheme.
In light of the fact that only $0.1\%$ of the total simulation budget is used, we could claim that our budget allocation problem and its solution strategy provide effective guidance in determining good budget allocation schemes.

\item[(2)] By comparing the difference between the maximum and minimum of the red curves, we can see that using an optimal budget allocation scheme could narrow a CI by at least 2 times.
When the total simulation budget is limited, solving the budget allocation problem is very beneficial.

\item[(3)] The best budget allocation schemes for VaR and CVaR estimation are quite similar. In particular, the optimal $N$ for constructing CI of VaR  and CVaR are around $2\times 10^3$. 

\item[(4)] It is worth mentioning  that the wider width for both VaR and CVaR estimation appears to be first decreasing in $N$, since the input uncertainty dominates the simulation error when $N$ is not large enough. When $N$ has approached certain level, stochastic uncertainty starts to play a more important role.
\end{itemize}

In conclusion, the simulation results for the sharing economy model provide empirical evidences for the importance and necessity of risk quantification in stochastic simulation under input uncertainty, as well as the advantages of solving the associated budget allocation problem for efficient nested simulation.
\color{black}
\section{Conclusion}\label{Sec6:Consusion}

In the present paper, we introduce risk quantification in stochastic simulation under input certainty, which rigorously quantifies extreme scenarios of mean response in all possible input models.
In particular, we propose nested Monte Carlo simulation to estimate VaR or CVaR of mean response w.r.t. input uncertainty.
We prove the asymptotical properties (consistency and normality) of the resultant nested risk estimators in different limiting senses under different sets of regularity conditions.
We further use the established properties to construct (asymptotically valid) CIs, and propose a practical framework of optimal budget allocation for improving the efficiency of nested risk simulation. 
\red{Lastly, we study a sharing economy example to illustrate the importance of accessing and controlling risk due to input uncertainty, and to demonstrate the effectiveness of our budget allocation scheme.}
The work in this paper can be viewed as a starting point of research on more general risk measures for risk quantification under input uncertainty.

On the other hand, the naive nested risk estimators considered here could be restrictive in risk quantification under input uncertainty for large-scale systems, due to the inefficiency of naive rare-event simulation.
The budget allocation problem solved in this paper partially addresses this issue in the sense that it leads to good outer versus inner sample size trade-off in reducing CI width.
Developing more sophisticated budget allocation schemes will be a promising direction of future research.

\appendix

\section*{APPENDIX}

\section{Proof of Theorem \ref{thm.3.1}}\label{ap.a}
For simplicity, let us use $\widehat{v}^{N}_\alpha$, $\widehat{c}^{N}_\alpha$, $\widetilde{v}^{N,M}_\alpha$, and $\widetilde{c}^{N,M}_\alpha$ to denote  $\widehat{v}_{\alpha}$, $\widehat{c}_{\alpha}$, $\widetilde{v}_{\alpha}$, and $\widetilde{c}_{\alpha}$, respectively. Therefore, we need to show that
 \begin{eqnarray*}
\lim\limits_{N\rightarrow \infty}\lim\limits_{M\rightarrow \infty}
\widetilde{v}^{N,M}_\alpha
= v_{\alpha},\quad w.p.1, \quad \mbox{and} \quad
\lim\limits_{N\rightarrow \infty}\lim\limits_{M\rightarrow \infty}
\widetilde{c}^{N,M}_\alpha
= c_{\alpha},\quad w.p.1.
\end{eqnarray*}
In view of the error decomposition
\begin{equation*}
\widetilde{v}^{N,M}_\alpha -v_{\alpha}=
\left(\widetilde{v}^{N,M}_\alpha-\widehat{v}^{N}_\alpha\right)+
\left(\widehat{v}^{N}_\alpha -v_{\alpha}\right)~
\mbox{and}~
\widetilde{c}^{N,M}_\alpha -c_{\alpha}=
\left(\widetilde{c}^{N,M}_\alpha-\widehat{c}^{N}_\alpha\right)+
\left(\widehat{c}^{N}_\alpha -c_{\alpha}\right),
\end{equation*}
it is sufficient to show that
\begin{equation}\label{eq.a.3}
\lim\limits_{N\rightarrow \infty} \left(\widehat{v}^{N}_\alpha-v_{\alpha}\right) =0, \quad w.p.1.\quad\mbox{and}\quad
\lim\limits_{N\rightarrow \infty}
\left(\widehat{c}^{N}_\alpha- c_{\alpha}\right) =0, \quad w.p.1.
\end{equation}
and for fixed $N$ and $\theta_1,...,\theta_N$,
\begin{equation}\label{eq.a.4}
\lim\limits_{M\rightarrow \infty} \left(\widetilde{v}^{N,M}_\alpha-\widehat{v}^{N}_\alpha\right) =0, \quad w.p.1.\quad\mbox{and}\quad
\lim\limits_{M\rightarrow \infty} \left(\widetilde{c}^{N,M}_\alpha-\widehat{c}^{N}_\alpha\right) =0, \quad w.p.1.
\end{equation}
To establish \eqref{eq.a.3}, we need the following lemma, and its proof can be found in online appendix.

\begin{lemma}\label{lem.a.1}
Under Assumption \ref{asm.3.1}.(ii),
\begin{eqnarray}
&&\left(\widehat{v}^{N}_\alpha-v_{\alpha}\right)=\frac{1}{f(v_\alpha)}
\left(\alpha-\frac{1}{N}\sum\limits_{i=1}^{N}
\mathds{1}\{H(\theta_{i})\le v_\alpha\}\right)+A_{N}, \label{eq.a.5}\\
&&\left(\widehat{c}^{N}_\alpha- c_{\alpha}\right)=\left(\frac{1}{N}\sum\limits_{i=1}^{N}
\left[v_{\alpha}+\frac{1}{1-\alpha}
\left(H(\theta_{i})-v_{\alpha}\right)^{+}\right]
-c_{\alpha}\right)+B_{N}, \label{eq.a.6}
\end{eqnarray}
where $A_{N}=O_{a.s.}(N^{-3/4}(\log N)^{3/4})$,  $B_{N}=O_{a.s.}(N^{-1}\log N)$. Here note that the statement $g(N)= O_{a.s.}(h(N))$ means that $g(N)\le C\cdot h(N)$ almost surely for some constant $C$.
\end{lemma}

Notice that $\frac{1}{N}\sum\limits_{i=1}^{N}
\mathds{1}\{H(\theta_{i})\le v_\alpha\}$ is an unbiased sample estimator of $\alpha$. By Strong Law of Large Numbers,
\begin{equation*}
\lim\limits_{N\rightarrow \infty} \frac{1}{N}\sum\limits_{i=1}^{N}
\mathds{1}\{H(\theta_{i})\le v_\alpha\}-\alpha =0, \quad w.p.1.
\end{equation*}
Combining with the fact $\lim\limits_{N\rightarrow \infty} A_N =0, \; w.p.1$, $\lim\limits_{N\rightarrow \infty} \left(\widehat{v}^{N}_\alpha-v_{\alpha}\right) =0, \; w.p.1.$  To show the latter half of \eqref{eq.a.3}, notice that $\frac{1}{N}\sum\limits_{i=1}^{N}
\left[v_{\alpha}+\frac{1}{1-\alpha}
\left(H(\theta_{i})-v_{\alpha}\right)^{+}\right]$ is an unbiased sample estimator of $c_\alpha$. Furthermore, by Assumption \ref{asm.3.1}.(i),
$$
\mathbb{E}[H^2(\theta)]=\mathbb{E}[\mathbb{E}^2[h(\theta;\xi)|\theta]]=
\int \mathbb{E}^2[h(\theta;\xi)|\theta] f(\theta)d\theta\le
\int \mathbb{E}[h^2(\theta;\xi)|\theta] f(\theta)d\theta<\infty.
$$
Therefore, $Var(H(\theta))$ is finite and $Var(v_{\alpha}+\frac{1}{1-\alpha}
\left(H(\theta)-v_{\alpha}\right)^{+})$ is also finite. By Strong Law of Large Numbers,
\begin{equation*}
\lim\limits_{N\rightarrow \infty} \frac{1}{N}\sum\limits_{i=1}^{N}
\left[v_{\alpha}+\frac{1}{1-\alpha}
\left(H(\theta_{i})-v_{\alpha}\right)^{+}\right]-c_\alpha =0, \quad w.p.1.
\end{equation*}
Combining with the fact $\lim\limits_{N\rightarrow \infty} B_N =0, \; w.p.1$, $\lim\limits_{N\rightarrow \infty} \left(\widehat{c}^{N}_\alpha-c_{\alpha}\right) =0, \; w.p.1$. \eqref{eq.a.3} has been established.

It remains to establish \eqref{eq.a.4} for fixed $N$ and scenarios $\theta_1,...,\theta_N$.
That is, we need to show for fixed $N$ and scenarios $\theta_1,...,\theta_N$,
\begin{eqnarray}
&& \lim\limits_{M\rightarrow \infty}
\widehat{H}_M(\theta^{(\alpha N)})-H(\theta_{(\alpha N)})=0, \quad w.p.1, \label{eq.a.9}\\
&& \lim\limits_{M\rightarrow \infty}
\left(\frac{1}{(1-\alpha)N}\sum\limits_{i= \alpha N}^{N}\widehat{H}_M(\theta^{(i)})-\frac{1}{(1-\alpha)N}
\sum\limits_{i= \alpha N}^{N}H(\theta_{(i)})\right)=0, \quad w.p.1. \label{eq.a.10}
\end{eqnarray}
Recall that for any $\theta_i, i=1,...,N$, $\mathbb{E}[h(\theta_i;\xi)|\theta_i]=H(\theta_{i})$ and  $Var[h(\theta_i;\xi)|\theta_i]=\tau^2_i<\infty$, where we use $\tau^2_i$ to denote $\tau^2_{\theta_{i}}$ with slight abuse of notations. By Strong Law of Large Numbers, we have for $i=1,...N$, $\widehat{H}_M(\theta_{i})\overset{M\rightarrow \infty}\rightarrow H(\theta_{i}), \; w.p.1$. Let $\Omega_i\subseteq \Omega$ be the set of such convergent scenarios for $i=1,...,N$, where $\Omega$ is the underlying sample space. Thus $P(\Omega_i)=1$. Denote $\bar{\Omega}:=\bigcap_{i=1}^{N} \Omega_i$, the intersection of all convergent scenario sets. Clearly, by Boole's Inequality $P(\bar{\Omega})=1$. Let us also denote, for any scenario $w\in\bar{\Omega}$, $\widehat{H}_M^w(\theta)$ as the sample realization of $\widehat{H}_M(\theta)$, $i=1,..,N$. Therefore, $\forall w\in\bar{\Omega}$
\begin{equation}\label{eq.a.11}
\lim\limits_{M\rightarrow \infty}(\widehat{H}^w_M(\theta_{1}),...,\widehat{H}^w_M(\theta_{N}))
=(H(\theta_{1}),...,H(\theta_{N})).
\end{equation}
\color{black}
Let $\epsilon:=\frac{1}{3}\min\{H(\theta_{i})-H(\theta_{j}):\;H(\theta_{i})\neq H(\theta_{j})\;
\; i\neq j,\; i,j=1,...,N\}$. By definition, \eqref{eq.a.11} implies that there exists a sufficient large $M_\epsilon$ such that $\forall M\ge M_\epsilon$, $|\widehat{H}^w_M(\theta_{i})-H(\theta_{i})|<\epsilon, i=1,...,N.$ It follows that, $\forall M\ge M_\epsilon$, and $i,j$ such that $H(\theta_{i})\neq H(\theta_{j})$,
$$\widehat{H}^w_M(\theta_{(i)})<\widehat{H}^w_M(\theta_{(j)})$$
If there exists $i<j$ such that $H(\theta_{i})= H(\theta_{j})$,  we simply write $\widehat{H}^w_M(\theta_{i})<\widehat{H}^w_M(\theta_{j})$ whatever their real order is.  We can make this assumption here because we only care about the order sequence of different $H$. No matter which one of these two estimates is larger, since both of them converge to the same $H$, the order sequence of $H$ does not change. \color{black} Then we have
$$
\widehat{H}^w_M(\theta_{(1)})<\widehat{H}^w_M(\theta_{(2)})<
\cdot\cdot\cdot<\widehat{H}^w_M(\theta_{(N)}).
$$
That is, $\forall M\ge M_\epsilon$, the sampling error so small that  the order sequence of the mean response is not perturbed. 

 Thus, $\forall M\ge M_\epsilon$, $(\theta_{w}^{(1)},..,\theta_{w}^{(N)})=(\theta_{(1)},...,\theta_{(N)})$,
where $\theta_{w}^{(i)}$ is the sample realization of $\theta^{(i)}$ with scenario $w$. Therefore, for any scenario $w\in\bar{\Omega}$,
$$
\lim\limits_{M\rightarrow \infty}
\widehat{H}^w_M(\theta_w^{(\alpha N)})=\lim\limits_{M\rightarrow \infty}
\widehat{H}^w_M(\theta_{(\alpha N)})=H(\theta_{(\alpha N)}),
$$
and
$$
\lim\limits_{M\rightarrow \infty}
\frac{1}{(1-\alpha)N}\sum\limits_{i= \alpha N}^{N}\widehat{H}^w_M(\theta_w^{(i)})= \lim\limits_{M\rightarrow \infty}
\frac{1}{(1-\alpha)N}\sum\limits_{i= \alpha N}^{N}\widehat{H}^w_M(\theta_{(i)}) =\frac{1}{(1-\alpha)N}
\sum\limits_{i= \alpha N}^{N}H(\theta_{(i)}).
$$
Notice $P(\bar{\Omega})=1$, \eqref{eq.a.9} and \eqref{eq.a.10} naturally hold.

\section{Proof of Theorem \ref{thm.3.2}}\label{ap.b}

Recall we need to show that
 \begin{eqnarray*}
\lim\limits_{N,M\rightarrow \infty}
\widetilde{v}^{N,M}_\alpha
= v_{\alpha},\quad w.p.1, \quad \mbox{and} \quad
\lim\limits_{N,M\rightarrow \infty}
\widetilde{c}^{N,M}_\alpha
= c_{\alpha},\quad w.p.1.
\end{eqnarray*}
In addition to the notations previously introduced in Appendix \ref{ap.a}, let us further use $\breve{v}_\alpha^{M}$ and $\breve{c}_\alpha^M$ to denote $v_\alpha(\widehat{H}_M(\theta))$ and $c_\alpha(\widehat{H}_M(\theta))$, respectively. That is, $\breve{v}_\alpha^{M}$ and $\breve{c}_\alpha^M$ are the exact $\alpha$-level VaR and CVaR of  noised mean response $\widehat{H}_M(\theta)$, respectively. As mentioned after Theorem \ref{thm.3.2}, in view of the fact that $\widetilde{v}_\alpha(H(\theta))=
\widehat{v}_\alpha(\widehat{H}_M(\theta))$ and $\widetilde{c}_\alpha(H(\theta))=
\widehat{c}_\alpha(\widehat{H}_M(\theta))$,
$\widetilde{v}^{N,M}_\alpha$ and $\widetilde{c}^{N,M}_\alpha$ could be regarded as the one-layer Monte Carlo estimator of $\breve{v}_\alpha^{M}$ and $\breve{c}_\alpha^M$, respectively. This observation inspires us to consider the following error decomposition
\begin{equation}\label{eq.b.2}
\widetilde{v}^{N,M}_\alpha -v_{\alpha}=
\left(\widetilde{v}^{N,M}_\alpha-\breve{v}^{M}_\alpha\right)+
\left(\breve{v}^{M}_\alpha -v_{\alpha}\right)
~
\mbox{and}~~
\widetilde{c}^{N,M}_\alpha -c_{\alpha}=
\left(\widetilde{c}^{N,M}_\alpha-\breve{c}^{M}_\alpha\right)+
\left(\breve{c}^{M}_\alpha -c_{\alpha}\right).
\end{equation}
Therefore, it is sufficient to show that
\begin{eqnarray}\label{eq.b.3}
\lim\limits_{M\rightarrow \infty} \breve{v}^{M}_\alpha =v_{\alpha} \quad\mbox{and}\quad
\lim\limits_{M\rightarrow \infty} \breve{c}^{M}_\alpha =c_{\alpha},
\end{eqnarray}
and uniformly for all $M$,
\begin{eqnarray}\label{eq.b.4}
\lim\limits_{N\rightarrow \infty} \widetilde{v}^{N,M}_\alpha = \breve{v}^{M}_\alpha \quad w.p.1 \quad\mbox{and}\quad
\lim\limits_{N\rightarrow \infty} \widetilde{c}^{N,M}_\alpha = \breve{c}^{M}_\alpha \quad w.p.1.
\end{eqnarray}
Let us first establish \eqref{eq.b.3}. The following lemmas will be useful, and we refer to online appendix for the proofs.
\begin{lemma}\label{lem.b.1}
Under Assumption \ref{asm.3.2}, if a sequence $t_M\rightarrow t$ as $M\rightarrow \infty$, then $\widetilde{f}_M(t_M)\rightarrow f(t)$ and $\widetilde{f}^\prime_M(t_M)\rightarrow f^\prime(t)$ as $M\rightarrow \infty$, where recall $\widetilde{f}_M(\cdot)$ is the p.d.f. of  noised mean response $\widehat{H}_M(\theta)$.
\end{lemma}

\begin{lemma}\label{lem.b.2}
Under Assumption \ref{asm.3.2},
\begin{equation*}
\breve{v}^{M}_\alpha =v_{\alpha}+\frac{-\Lambda^\prime(v_{\alpha})}
{M f(v_{\alpha})} + o_M(\frac{1}{M}),
\end{equation*}
where the function $\Lambda(t)=1/2 f(t)\mathbb{E}[\tau_\theta^2|H(\theta)=t]$ and $o_M(\frac{1}{M})$ means this quantity goes to zero faster than $\frac{1}{M}$ (almost surely).
\end{lemma}

\begin{lemma}\label{lem.b.3}
Under Assumption \ref{asm.3.2},
\begin{equation}\label{eq.b.13}
\breve{c}^{M}_\alpha =c_{\alpha}+\frac{\Lambda(v_{\alpha})}
{(1-\alpha)M } + o_M(\frac{1}{M}).
\end{equation}
\end{lemma}

\begin{lemma}\label{lem.b.4}
Under Assumption \ref{asm.3.2},
\begin{equation}
\left(\widetilde{v}^{N,M}_\alpha- \breve{v}^{M}_\alpha\right)=\frac{1}{\widetilde{f}(\breve{v}^M_\alpha)}
\left(\alpha-\frac{1}{N}\sum\limits_{i=1}^{N}
\mathds{1}\{\widehat{H}_M(\theta_{i})\le \breve{v}^M_\alpha\}\right)+O_{a.s.}(N^{-3/4}(\log N)^{3/4}), \label{eq.b.20}
\end{equation}
\begin{equation}\label{eq.b.21}
\small{\left(\widetilde{c}^{N,M}_\alpha- \breve{c}^{M}_\alpha\right)=\left(\frac{1}{N}\sum\limits_{i=1}^{N}
\left[\breve{v}^M_\alpha+\frac{1}{1-\alpha}
\left(\widehat{H}_M(\theta_{i})-\breve{v}^M_\alpha\right)^{+}\right]
-\breve{c}^{M}_\alpha\right)+O_{a.s.}(N^{-1}\log N),}
\end{equation}
where $O_{a.s.}(N^{-3/4}(\log N)^{3/4})$ and $O_{a.s.}(N^{-1}\log N)$ hold uniformly for all $M$.
\end{lemma}

By Lemma \ref{lem.b.2} and Lemma \ref{lem.b.3}, \eqref{eq.b.3} naturally holds. Furthermore, Lemma \ref{lem.b.4} implies \eqref{eq.b.4}.

\section{Proof of Theorem \ref{thm.3.4}}\label{ap.d}

Follow the notations in Appendix \ref{ap.a} and \ref{ap.b}, we need to show that under Assumption \ref{asm.3.2}, the existence of the limit $K^2=\lim\limits_{N, M\rightarrow \infty}N/M^2$ is a sufficient and necessary condition for
\begin{eqnarray}
&& \lim\limits_{N,M\rightarrow \infty}\sqrt{N}\left(\widetilde{v}^{N,M}_\alpha-v_{\alpha}\right)
\overset{\mathcal{D}}=
\sigma_{v}\mathcal{N}(0,1)+|K|\mu_v,\label{eq.d.1}\\
&&\lim\limits_{N,M\rightarrow \infty}\sqrt{N}\left(\widetilde{c}^{N,M}_\alpha-c_{\alpha}\right)
\overset{\mathcal{D}}=
\sigma_{c}\mathcal{N}(0,1)+|K|\mu_c. \label{eq.d.2}
\end{eqnarray}
By Lemma \ref{lem.b.2}, \ref{lem.b.3} and \ref{lem.b.4}, we have
\begin{eqnarray*}
\left(\widetilde{v}^{N,M}_\alpha-v_\alpha\right)
=Err_1+Bias_1&=&\frac{1}{\widetilde{f}(\breve{v}^M_\alpha)}
\left(\alpha-\frac{1}{N}\sum\limits_{i=1}^{N}
\mathds{1}\{\widehat{H}_M(\theta_{i})\le \breve{v}^M_\alpha\}\right)\nonumber \\ &&+\frac{-\Lambda^\prime(v_{\alpha})}
{M f(v_{\alpha})} + o_M(\frac{1}{M})+O_{a.s.}(N^{-3/4}(\log N)^{3/4}), \label{eq.d.3}\\
\left(\widetilde{c}^{N,M}_\alpha-c_\alpha\right)
=Err_2+Bias_2&=&\left(\frac{1}{N}\sum\limits_{i=1}^{N}
\left[\breve{v}^M_\alpha+\frac{1}{1-\alpha}
\left(\widehat{H}_M(\theta_{i})-\breve{v}^M_\alpha\right)^{+}\right]
-\breve{c}^{M}_\alpha\right)\nonumber\\
&&+\frac{\Lambda(v_{\alpha})}
{(1-\alpha)M } + o_M(\frac{1}{M})+O_{a.s.}(N^{-1}\log N). \label{eq.d.4}
\end{eqnarray*}
Note that $\breve{v}^M_\alpha \rightarrow v_\alpha$, $\breve{c}^M_\alpha \rightarrow c_\alpha$, and $\widetilde{f}(\breve{v}^M_\alpha)\rightarrow f(v_\alpha)$ as $M\rightarrow\infty$. We have
$$
\lim\limits_{M\rightarrow \infty}
\frac{1}{\widetilde{f}(\breve{v}^M_\alpha)} \left(\alpha-
\mathds{1}\{\widehat{H}_M(\theta)\le \breve{v}^M_\alpha\}\right)
=\frac{1}{f(v_\alpha)} \left(\alpha- \mathds{1}\{H(\theta)\le v_\alpha\}\right),\quad w.p.1,
$$
and
$$
\lim\limits_{M\rightarrow \infty}
\left(\breve{v}^M_\alpha+\frac{1}{1-\alpha}
\left(\widehat{H}_M(\theta)-\breve{v}^M_\alpha\right)^{+}
-\breve{c}^{M}_\alpha\right)=
\left(v_\alpha+\frac{1}{1-\alpha}
\left(H(\theta)-v_\alpha\right)^{+}
-c_\alpha\right),\quad w.p.1.
$$
What's more,
\begin{eqnarray*}
\sigma^2_{v}&=&Var\left[\frac{1}{f(v_\alpha)}
\left(\alpha-
\mathds{1}\{H(\theta_{i})\le v_\alpha\}\right)\right]\\
&=&\frac{1}{f^2(v_\alpha)} Var\left[\mathds{1}\{H(\theta)\le v_\alpha\}\right]= \frac{\alpha(1-\alpha)}{f^2(v_\alpha)},
\end{eqnarray*}
and
\begin{eqnarray*}
\sigma^2_{c}=Var\left[
\left[v_{\alpha}+\frac{1}{1-\alpha}
\left(H(\theta_{i})-v_{\alpha}\right)^{+}\right]
-c_{\alpha}\right]=\frac{1}{(1-\alpha)^2} Var\left[\left(H(\theta)
-v_{\alpha}\right)^{+}\right].
\end{eqnarray*}
Therefore, by Central Limit Theorem, \eqref{eq.d.1} and \eqref{eq.d.2} hold if and only if $K^2=\lim\limits_{N, M\rightarrow \infty}N/M^2$ exists.

\section{Proof of Theorem \ref{thm.3.5}}\label{ap.e}
When $N=o(M^2)$, the bias terms will be asymptotically insignificant compared with the $O(\frac{1}{\sqrt{N}})$ error term. The CI in \eqref{eq.3.25} and \eqref{eq.3.26} will become
\begin{eqnarray}
\left[\widetilde{v}_{\alpha}+
\frac{t_{\beta/2,N-1}\widehat{\sigma}_{v}}{\sqrt{N}},~~ \widetilde{v}_{\alpha}+
\frac{t_{1-\beta/2,N-1}\widehat{\sigma}_{v}}{\sqrt{N}}\right]
\end{eqnarray}
and
\begin{eqnarray}
\left[
\widetilde{c}_{\alpha}+
\frac{t_{\beta/2,N-1}\widehat{\sigma}_{c}}{\sqrt{N}},~~
\widetilde{c}_{\alpha}+
\frac{t_{1-\beta/2,N-1}\widehat{\sigma}_{c}}{\sqrt{N}}\right].
\end{eqnarray}
So we just need to show the following limits
$$
\left\{\begin{array}{l}
\lim\limits_{N\rightarrow \infty}\lim\limits_{M\rightarrow \infty} P\{|Err_1|\le 2\frac{t_{1-\beta/2,N-1}\widehat{\sigma}_{v}}{\sqrt{N}}\}=1-\beta,\\
\lim\limits_{N\rightarrow \infty}\lim\limits_{M\rightarrow \infty} P\{|Err_2|\le2\frac{t_{1-\beta/2,N-1}\widehat{\sigma}_{c}}{\sqrt{N}}\}
=1-\beta.
\end{array}\right.
$$
where recall that $Err_1$, $Err_2$ are defined in \eqref{eq.3.12} and \eqref{eq.3.13}. In view of the fact that a Student's t-distribution converges to a standard normal distribution as the degree of freedom goes to infinity, the almost sure convergence of variance estimators by Strong Law of Large Numbers, and the consistency of kernel density estimation, these limits naturally hold following Theorem \ref{thm.3.4}.



\bibliographystyle{ACM-Reference-Format-Journals}
\bibliography{Zhou-Bibtex}


\begin{thebibliography}{00}


\ifx \showCODEN    \undefined \def \showCODEN     #1{\unskip}     \fi
\ifx \showDOI      \undefined \def \showDOI       #1{{\tt DOI:}\penalty0{#1}\ }
  \fi
\ifx \showISBNx    \undefined \def \showISBNx     #1{\unskip}     \fi
\ifx \showISBNxiii \undefined \def \showISBNxiii  #1{\unskip}     \fi
\ifx \showISSN     \undefined \def \showISSN      #1{\unskip}     \fi
\ifx \showLCCN     \undefined \def \showLCCN      #1{\unskip}     \fi
\ifx \shownote     \undefined \def \shownote      #1{#1}          \fi
\ifx \showarticletitle \undefined \def \showarticletitle #1{#1}   \fi
\ifx \showURL      \undefined \def \showURL       #1{#1}          \fi

\bibitem[\protect\citeauthoryear{Ankenman, Nelson, and Staum}{Ankenman
  et~al\mbox{.}}{2010}]%
        {ankenman2010stochastic}
{Bruce Ankenman}, {Barry~L Nelson}, {and} {Jeremy Staum}. 2010.
\newblock \showarticletitle{Stochastic kriging for simulation metamodeling}.
\newblock {\em Operations Research\/} {58}, 2 (2010), 371--382.
\newblock


\bibitem[\protect\citeauthoryear{Artzner, Delbaen, Eber, and Heath}{Artzner
  et~al\mbox{.}}{1999}]%
        {artzner1997coherent}
{Philippe Artzner}, {Freddy Delbaen}, {Jean-Marc Eber}, {and} {David Heath}.
  1999.
\newblock \showarticletitle{Coherent measures of risk}.
\newblock {\em Mathematical Finance\/}  {9} (1999), 203--228.
\newblock


\bibitem[\protect\citeauthoryear{Banerjee, Riquelme, and Johari}{Banerjee
  et~al\mbox{.}}{2015}]%
        {banerjee2015pricing}
{Siddhartha Banerjee}, {Carlos Riquelme}, {and} {Ramesh Johari}. 2015.
\newblock \showarticletitle{Pricing in ride-share platforms: A
  queueing-theoretic approach}.
\newblock  (2015).
\newblock


\bibitem[\protect\citeauthoryear{Barton}{Barton}{2012}]%
        {barton2012tutorial}
{Russell~R Barton}. 2012.
\newblock \showarticletitle{Tutorial: input uncertainty in outout analysis}. In
  {\em Proceedings of the 2012 Winter Conference on Simulation}. IEEE, Berlin,
  Germany, 1--12.
\newblock


\bibitem[\protect\citeauthoryear{Barton, Nelson, and Xie}{Barton
  et~al\mbox{.}}{2013}]%
        {barton2013quantifying}
{Russell~R Barton}, {Barry~L Nelson}, {and} {Wei Xie}. 2013.
\newblock \showarticletitle{Quantifying input uncertainty via simulation
  confidence intervals}.
\newblock {\em INFORMS Journal on Computing\/} {26}, 1 (2013), 74--87.
\newblock


\bibitem[\protect\citeauthoryear{Barton and Schruben}{Barton and
  Schruben}{1993}]%
        {barton1993uniform}
{Russell~R Barton} {and} {Lee~W Schruben}. 1993.
\newblock \showarticletitle{Uniform and bootstrap resampling of empirical
  distributions}. In {\em Proceedings of the 1993 Winter Conference on
  Simulation}. IEEE, Los Angeles, CA, 503--508.
\newblock


\bibitem[\protect\citeauthoryear{Barton and Schruben}{Barton and
  Schruben}{2001}]%
        {barton2001resampling}
{Russell~R Barton} {and} {Lee~W Schruben}. 2001.
\newblock \showarticletitle{Resampling methods for input modeling}. In {\em
  Proceedings of the 2001 Winter Conference on Simulation}. IEEE, Arlington,
  VA, 372--378.
\newblock


\bibitem[\protect\citeauthoryear{Biller and Corlu}{Biller and Corlu}{2011}]%
        {biller2011accounting}
{Bahar Biller} {and} {Canan~G Corlu}. 2011.
\newblock \showarticletitle{Accounting for parameter uncertainty in large-scale
  stochastic simulations with correlated inputs}.
\newblock {\em Operations Research\/} {59}, 3 (2011), 661--673.
\newblock


\bibitem[\protect\citeauthoryear{Broadie, Du, and Moallemi}{Broadie
  et~al\mbox{.}}{2011}]%
        {broadie2011efficient}
{Mark Broadie}, {Yiping Du}, {and} {Ciamac~C Moallemi}. 2011.
\newblock \showarticletitle{Efficient risk estimation via nested sequential
  simulation}.
\newblock {\em Management Science\/} {57}, 6 (2011), 1172--1194.
\newblock


\bibitem[\protect\citeauthoryear{Cheng and Holloand}{Cheng and
  Holloand}{1997}]%
        {cheng1997sensitivity}
{Russell~CH Cheng} {and} {Wayne Holloand}. 1997.
\newblock \showarticletitle{Sensitivity of computer simulation experiments to
  errors in input data}.
\newblock {\em Journal of Statistical Computation and Simulation\/} {57}, 1-4
  (1997), 219--241.
\newblock


\bibitem[\protect\citeauthoryear{Chick}{Chick}{2001}]%
        {chick2001input}
{Stephen~E Chick}. 2001.
\newblock \showarticletitle{Input distribution selection for simulation
  experiments: accounting for input uncertainty}.
\newblock {\em Operations Research\/} {49}, 5 (2001), 744--758.
\newblock


\bibitem[\protect\citeauthoryear{Degeilh and Gross}{Degeilh and Gross}{2015}]%
        {degeilhstochastic}
{Yannick Degeilh} {and} {George Gross}. 2015.
\newblock \showarticletitle{Stochastic simulation of utility-scale storage
  resources in power systems with integrated renewable resources}.
\newblock {\em IEEE Transactions on Power Systems\/} {30}, 3 (2015),
  1424--1434.
\newblock


\bibitem[\protect\citeauthoryear{Glasserman, Heidelberger, and
  Shahabuddin}{Glasserman et~al\mbox{.}}{2000}]%
        {glasserman2000variance}
{Paul Glasserman}, {Philip Heidelberger}, {and} {Perwez Shahabuddin}. 2000.
\newblock \showarticletitle{Variance reduction techniques for estimating
  value-at-risk}.
\newblock {\em Management Science\/} {46}, 10 (2000), 1349--1364.
\newblock


\bibitem[\protect\citeauthoryear{Gordy and Juneja}{Gordy and Juneja}{2010}]%
        {gordy2010nested}
{Michael~B Gordy} {and} {Sandeep Juneja}. 2010.
\newblock \showarticletitle{Nested simulation in portfolio risk measurement}.
\newblock {\em Management Science\/} {56}, 10 (2010), 1833--1848.
\newblock


\bibitem[\protect\citeauthoryear{Henderson}{Henderson}{2003}]%
        {henderson2003}
{Shane~G. Henderson}. 2003.
\newblock \showarticletitle{Input model uncertainty: Why do we care and what
  should we do about it?}. In {\em Proceedings of the 2003 Winter Conference on
  Simulation}. IEEE, New Orleans, LA, 90--100.
\newblock


\bibitem[\protect\citeauthoryear{Hong, Hu, and Liu}{Hong et~al\mbox{.}}{2014}]%
        {hong2014monte}
{L~Jeff Hong}, {Zhaolin Hu}, {and} {Guangwu Liu}. 2014.
\newblock \showarticletitle{Monte Carlo methods for value-at-risk and
  conditional value-at-risk: a review}.
\newblock {\em ACM Transactions on Modeling and Computer Simulation\/} {24}, 4
  (2014), 1--37.
\newblock


\bibitem[\protect\citeauthoryear{Lan, Nelson, and Staum}{Lan
  et~al\mbox{.}}{2010}]%
        {lan2010confidence}
{Hai Lan}, {Barry~L Nelson}, {and} {Jeremy Staum}. 2010.
\newblock \showarticletitle{A confidence interval procedure for expected
  shortfall risk measurement via two-level simulation}.
\newblock {\em Operations Research\/} {58}, 5 (2010), 1481--1490.
\newblock


\bibitem[\protect\citeauthoryear{Lee}{Lee}{1998}]%
        {lee1998monte}
{Shing-Hoi Lee}. 1998.
\newblock {\em Monte Carlo Computation of Conditional Expectation Quantiles}.
\newblock Ph.D. Dissertation. Stanford University, Stanford, USA.
\newblock


\bibitem[\protect\citeauthoryear{Liu and Staum}{Liu and Staum}{2010}]%
        {liu2010stochastic}
{Ming Liu} {and} {Jeremy Staum}. 2010.
\newblock \showarticletitle{Stochastic kriging for efficient nested simulation
  of expected shortfall}.
\newblock {\em Journal of Risk\/} {12}, 3 (2010), 3–--27.
\newblock


\bibitem[\protect\citeauthoryear{Rockafellar and Uryasev}{Rockafellar and
  Uryasev}{2000}]%
        {rockafellar2000optimization}
{R~Tyrrell Rockafellar} {and} {Stanislav Uryasev}. 2000.
\newblock \showarticletitle{Optimization of conditional value-at-risk}.
\newblock {\em Journal of Risk\/}  {2} (2000), 21--42.
\newblock


\bibitem[\protect\citeauthoryear{Rouvinez}{Rouvinez}{1997}]%
        {rouvinez1997}
{Christophe Rouvinez}. 1997.
\newblock \showarticletitle{Going greek with var}.
\newblock {\em Risk\/} {10}, 2 (1997), 57--63.
\newblock


\bibitem[\protect\citeauthoryear{Serfling}{Serfling}{2009}]%
        {serfling2009approximation}
{Robert~J Serfling}. 2009.
\newblock {\em Approximation Theorems of Mathematical Statistics}. Vol. 162.
\newblock John Wiley \& Sons.
\newblock


\bibitem[\protect\citeauthoryear{Song and Nelson}{Song and Nelson}{2015}]%
        {song2015quickly}
{Eunhye Song} {and} {Barry~L Nelson}. 2015.
\newblock \showarticletitle{Quickly assessing contributions to input
  uncertainty}.
\newblock {\em IIE Transactions\/}  {47} (2015), 893--909.
\newblock


\bibitem[\protect\citeauthoryear{Steckley}{Steckley}{2006}]%
        {steckley2006}
{G. Steckley}. 2006.
\newblock {\em Estimating the Density of A Conditional Expectation}.
\newblock Ph.D. Dissertation. Cornell University, Ithaca, NY.
\newblock


\bibitem[\protect\citeauthoryear{Sun and Hong}{Sun and Hong}{2010}]%
        {sun2010asymptotic}
{Lihua Sun} {and} {L~Jeff Hong}. 2010.
\newblock \showarticletitle{Asymptotic representations for importance-sampling
  estimators of value-at-risk and conditional value-at-risk}.
\newblock {\em Operations Research Letters\/} {38}, 4 (2010), 246--251.
\newblock


\bibitem[\protect\citeauthoryear{Sun, Apley, and Staum}{Sun
  et~al\mbox{.}}{2011}]%
        {sun2011efficient}
{Yunpeng Sun}, {Daniel~W Apley}, {and} {Jeremy Staum}. 2011.
\newblock \showarticletitle{Efficient nested simulation for estimating the
  variance of a conditional expectation}.
\newblock {\em Operations Research\/} {59}, 4 (2011), 998--1007.
\newblock


\bibitem[\protect\citeauthoryear{Xie, Nelson, and Barton}{Xie
  et~al\mbox{.}}{2014}]%
        {xie2014bayesian}
{Wei Xie}, {Barry~L Nelson}, {and} {Russell~R Barton}. 2014.
\newblock \showarticletitle{A Bayesian framework for quantifying uncertainty in
  stochastic simulation}.
\newblock {\em Operations Research\/} {62}, 6 (2014), 1439--1452.
\newblock


\bibitem[\protect\citeauthoryear{Xie, Nelson, and Barton}{Xie
  et~al\mbox{.}}{2015}]%
        {xie2014statistical}
{Wei Xie}, {Barry~L Nelson}, {and} {Russell~R Barton}. 2015.
\newblock {\em Statistical uncertainty analysis for stochastic simulation}.
\newblock {T}echnical {R}eport. Rensselaer Polytechnic Institute, New York,
  USA.
\newblock


\bibitem[\protect\citeauthoryear{Zouaoui and Wilson}{Zouaoui and
  Wilson}{2003}]%
        {zouaoui2003accounting}
{Faker Zouaoui} {and} {James~R Wilson}. 2003.
\newblock \showarticletitle{Accounting for parameter uncertainty in simulation
  input modeling}.
\newblock {\em IIE Transactions\/} {35}, 9 (2003), 781--792.
\newblock


\bibitem[\protect\citeauthoryear{Zouaoui and Wilson}{Zouaoui and
  Wilson}{2004}]%
        {zouaoui2004accounting}
{Faker Zouaoui} {and} {James~R Wilson}. 2004.
\newblock \showarticletitle{Accounting for input-model and input-parameter
  uncertainties in simulation}.
\newblock {\em IIE Transactions\/} {36}, 11 (2004), 1135--1151.
\newblock


\end{thebibliography}


\elecappendix
\textbf{Proof of Lemma \ref{lem.a.1}.}
\begin{proof}
The asymptotical representation \eqref{eq.a.5} is exactly Theorem 2.5.1 in \cite{serfling2009approximation} under Assumption \ref{asm.3.1}.(ii). The asymptotical representation \eqref{eq.a.6} is the special case of Theorem 2 in \cite{sun2010asymptotic}, when the importance sampling measure $\mathcal{L}\equiv 1$.
\end{proof}

\textbf{Proof of Lemma \ref{lem.b.1}.}
\begin{proof}
This result is exactly Lemma 1 in \cite{gordy2010nested}. For convenience, we will briefly present the proof. Recall that $\widehat{H}_M(\theta)=H(\theta)+\bar{\mathcal{E}}_M/\sqrt{M}$, where $(H(\theta), \bar{\mathcal{E}}_M)$ has a joint distribution $p_M(h, e)$. Therefore,
$$
\widetilde{f}_M(t_M)=\int_{\mathbb{R}} p_M(t_M-e/\sqrt{M}, e)de\quad
\mbox{and}\quad
f(t)=\int_{\mathbb{R}} p_M(t, e)de.
$$
It follows that
$$
\widetilde{f}_M(t_M)-f(t)=\int_{\mathbb{R}}\left( p_M(t-e/\sqrt{M}, e)-p_M(t, e)\right)de.
$$
By Taylor series expansion, this equals
$$
(t_M-t)\int_{\mathbb{R}} \frac{\partial}{\partial t} p_M(\check{t}_M, e)de - \frac{1}{\sqrt{M}} \int_{\mathbb{R}}e \frac{\partial}{\partial t} p_M(\check{t}_M, e)de,
$$
where $\check{t}_M$ lives in between $t_M$ and $t$. By Assumption 1 and the fact that $t_M\rightarrow t$ as $M\rightarrow \infty$, both terms converge to zero as $M\rightarrow \infty$.
\end{proof}

\textbf{Proof of Lemma \ref{lem.b.2}.}
\begin{proof}
This result is very similar to Proposition 1 in \cite{gordy2010nested}. The proof here will mainly follow \cite{gordy2010nested}'s proof.

Recall that $\widetilde{F}_M(\cdot)$ is the c.d.f. of the noised mean response $\widehat{H}_M(\theta)$, and $\breve{v}_\alpha^{M}$ is the exact $\alpha$-level VaR of $\widehat{H}_M(\theta)$. Thus, $\widetilde{F}_M(\breve{v}_\alpha^{M})=\alpha$. By Taylor expansion, we have
\begin{equation*}
\alpha=\widetilde{F}_M(\breve{v}_\alpha^{M})=
\widetilde{F}_M(v_\alpha)+(\breve{v}_\alpha^{M}-v_\alpha) \widetilde{f}_M(v_\alpha)+\frac{(\breve{v}_\alpha^{M}-v_\alpha)^2}{2}
\widetilde{f}^\prime_M(\check{v}^M_\alpha),
\end{equation*}
where $\check{v}^M_\alpha$ lives in between $\breve{v}^M_\alpha$ and $v_\alpha$. Therefore,
\begin{equation}\label{eq.b.6}
\alpha-
\widetilde{F}_M(v_\alpha)=(\breve{v}_\alpha^{M}-v_\alpha) \widetilde{f}_M(v_\alpha)+\frac{(\breve{v}_\alpha^{M}-v_\alpha)^2}{2}
\widetilde{f}_M^\prime(\check{v}^M_\alpha),
\end{equation}
Furthermore, notice that
\begin{equation}\label{eq.b.7}
\widetilde{F}_M(v_\alpha)=\int_{-\infty}^{v_\alpha}
\widetilde{f}_M(t)dt= \int_{\mathbb{R}}\int_{-\infty}^{v_\alpha-e/\sqrt{M}} p_M(t,e)dtde,
\end{equation}
and
\begin{equation}\label{eq.b.8}
\alpha=F(v_\alpha)=\int_{-\infty}^{v_\alpha} f(t)dt =
\int_{\mathbb{R}}\int_{-\infty}^{v_\alpha} p_M(t,e)dtde.
\end{equation}
Combining \eqref{eq.b.7} and \eqref{eq.b.8}, we have
\begin{equation}\label{eq.b.9}
\alpha-\widetilde{F}_M(v_\alpha)=
\int_{\mathbb{R}}\int_{v_\alpha-e/\sqrt{M}}^{v_\alpha} p_M(t,e)dtde.
\end{equation}
By Taylor expansion, we have
\begin{equation*}
p_M(t,e)=p_M(v_\alpha, e)+ (t-v_\alpha)\frac{\partial}{\partial t} p_M(v_\alpha, e) +\frac{(t-v_\alpha)^2}{2}\frac{\partial^2}{\partial t^2} p_M(\check{v}_\alpha, e),
\end{equation*}
where $\check{v}_\alpha$ lives in between $v_\alpha$ and $t$. Hence,
\begin{eqnarray}\label{eq.b.11}
\alpha-\widetilde{F}_M(v_\alpha)&=& \int_{\mathbb{R}}\int_{v_\alpha-e/\sqrt{M}}^{v_\alpha} p_M(v_\alpha, e)dtde + \int_{\mathbb{R}}\int_{v_\alpha-e/\sqrt{M}}^{v_\alpha} (t-v_\alpha)\frac{\partial}{\partial t} p_M(v_\alpha, e)dt de \nonumber \\
&&+\int_{\mathbb{R}}\int_{v_\alpha-e/\sqrt{M}}^{v_\alpha}  \frac{(t-v_\alpha)^2}{2}\frac{\partial^2}{\partial t^2} p_M(\check{v}_\alpha, e) dtde.
\end{eqnarray}
The first term of the right hand side of \eqref{eq.b.11} satisfies
$$
\int_{\mathbb{R}}\int_{v_\alpha-e/\sqrt{M}}^{v_\alpha} p_M(v_\alpha, e)dtde =\int_{\mathbb{R}}\frac{e}{\sqrt{M}} p_M(v_\alpha, e)de=
\frac{f(v_\alpha)}{\sqrt{M}} \mathbb{E}[\bar{\mathcal{E}}_M|H(\theta)=v_\alpha]=0.
$$
The second term of \eqref{eq.b.11} satisfies
\begin{eqnarray*}
\int_{\mathbb{R}}\int_{v_\alpha-e/\sqrt{M}}^{v_\alpha} (t-v_\alpha)\frac{\partial}{\partial t} p_M(v_\alpha, e)dt de
&=&-\frac{1}{2M} \int_{\mathbb{R}} e^2\frac{\partial}{\partial t} p_M(v_\alpha, e) de\\
&=& -\frac{1}{2M} \frac{\partial}{\partial t} \int_{\mathbb{R}} e^2
p_M(v_\alpha, e) de\\
&=& -\frac{1}{2M} \frac{\partial}{\partial t} f(v_\alpha)\mathbb{E}
[\tau_\theta^2| H(\theta)=v_\alpha]\\
&=&-\frac{1}{M} \Lambda^{\prime}(v_\alpha).
\end{eqnarray*}
By Assumption \ref{asm.3.2}, the third term of \eqref{eq.b.11} is in the order of $O_M(M^{-\frac{3}{2}})$.
Therefore,
\begin{equation}\label{eq.b.12}
\alpha-\widetilde{F}_M(v_\alpha)=-\frac{1}{M} \Lambda^{\prime}(v_\alpha)+ O_M(M^{-\frac{3}{2}}).
\end{equation}
Combining \eqref{eq.b.12} with \eqref{eq.b.6}, we have
\begin{equation*}
(\breve{v}_\alpha^{M}-v_\alpha) \widetilde{f}_M(v_\alpha)+\frac{(\breve{v}_\alpha^{M}-v_\alpha)^2}{2}
\widetilde{f}_M^\prime(\check{v}^M_\alpha)=
-\frac{1}{M} \Lambda^{\prime}(v_\alpha)+ O_M(M^{-\frac{3}{2}}),
\end{equation*}
where note that by Assumption \ref{asm.3.2}, it is easy to see that $\widetilde{f}_M^\prime(t)$ is uniformly bounded for all $t$ and $M$. Combining with Lemma \ref{lem.b.1}, Lemma \ref{lem.b.2} holds.
\end{proof}

\textbf{Proof of Lemma \ref{lem.b.3}.}
\begin{proof}
The result here is very similar to Proposition 3 in \cite{gordy2010nested}, and our proof will mainly follow \cite{gordy2010nested}'s proof.
Note that by Mean Value Theorem,
\begin{eqnarray*}
\breve{c}^{M}_\alpha &=&\frac{1}{1-\alpha}\mathbb{E}\left[
\widetilde{H}_M(\theta)\cdot \mathds{1}\{\widetilde{H}_M(\theta)\ge \breve{v}_\alpha^{M}\}\right]=\frac{1}{1-\alpha} \int_{\breve{v}^M_\alpha}^{\infty}t\widetilde{f}_M(t)dt\\
&=&\frac{1}{1-\alpha} \int_{v_\alpha}^{\infty}t\widetilde{f}_M(t)dt+
\frac{1}{1-\alpha} \int^{v_\alpha}_{\breve{v}^M_\alpha}t\widetilde{f}_M(t)dt\\
&=&\frac{1}{1-\alpha}\mathbb{E}\left[
\widetilde{H}_M(\theta)\cdot \mathds{1}\{\widetilde{H}_M(\theta)\ge v_\alpha\}\right]+ \frac{1}{1-\alpha} (v_\alpha-\breve{v}^M_\alpha)
t_v \widetilde{f}_M(t_v),
\end{eqnarray*}
where $t_v$ lives in between $v_\alpha$ and $\breve{v}^M_\alpha$. By
Lemma \ref{lem.b.2}, we know
\begin{equation*}
\frac{1}{1-\alpha} (v_\alpha-\breve{v}^M_\alpha)
t_v \widetilde{f}_M(t_v)=\frac{v_\alpha \Lambda^{\prime}(v_\alpha)}{(1-\alpha)M}+o_M(\frac{1}{M}).
\end{equation*}
Therefore,
\begin{eqnarray*}\label{eq.b.14}
\breve{c}^{M}_\alpha &=&\frac{1}{1-\alpha}\mathbb{E}\left[
\widetilde{H}_M(\theta)\cdot \mathds{1}\{\widetilde{H}_M(\theta)\ge v_\alpha\}\right]+\frac{v_\alpha \Lambda^{\prime}(v_\alpha)}{(1-\alpha)M}+o_M(\frac{1}{M}).
\end{eqnarray*}
Further notice that
\begin{eqnarray*}
&&\frac{1}{1-\alpha}\mathbb{E}\left[
\widetilde{H}_M(\theta)\cdot \mathds{1}\{\widetilde{H}_M(\theta)\ge v_\alpha\}\right]=\frac{1}{1-\alpha}
\int_{\mathbb{R}}\int_{v_\alpha-e/\sqrt{M}}^{\infty}
(t+e/\sqrt{M}) p_M(t,e)dtde, \label{eq.b.15}\\
&&c_\alpha=\frac{1}{1-\alpha}\mathbb{E}\left[H(\theta)\cdot \mathds{1}\{H(\theta)\ge v_\alpha\}\right]=\frac{1}{1-\alpha}
\int_{\mathbb{R}}\int_{v_\alpha}^{\infty}t p_M(t,e)dtde, \label{eq.b.16}
\end{eqnarray*}
and
$$
\int_{\mathbb{R}}\int_{v_\alpha}^{\infty} e p_M(t,e)dtde=\int_{v_\alpha}^{\infty}\mathbb{E}[\tilde{\mathcal{E}}_M|
H(\theta)=t] f(t)dt=0.
$$
Therefore,
\begin{eqnarray}
\breve{c}^{M}_\alpha-c_\alpha&=&\frac{1}{1-\alpha}\left(
\int_{\mathbb{R}}\int_{v_\alpha-e/\sqrt{M}}^{v_\alpha}t p_M(t,e)dtde+
\frac{1}{\sqrt{N}}\int_{\mathbb{R}} e \int_{v_\alpha-e/\sqrt{M}}^{v_\alpha} p_M(t,e)dtde\right)\nonumber\\
&&+\frac{v_\alpha\Lambda^{\prime}(v_\alpha)}{(1-\alpha)M}
+o_M(\frac{1}{M}).\label{eq.b.17}
\end{eqnarray}
Similar to the derivation (by taking Taylor expansion) from \eqref{eq.b.9} to \eqref{eq.b.12}, we have
\begin{eqnarray}\label{eq.b.18}
\frac{1}{1-\alpha}\int_{\mathbb{R}}\int_{v_\alpha-e/\sqrt{M}}^{v_\alpha}
t p_M(t,e)dtde=-\frac{\Lambda(v_\alpha)}{(1-\alpha)M}
-\frac{v_\alpha \Lambda^{\prime}(v_\alpha)}{(1-\alpha)M}+O_M(M^{-\frac{3}{2}}),
\end{eqnarray}
and
\begin{eqnarray}\label{eq.b.19}
\frac{1}{1-\alpha}\frac{1}{\sqrt{N}}\int_{\mathbb{R}} e \int_{v_\alpha-e/\sqrt{M}}^{v_\alpha} p_M(t,e)dtde=2\frac{\Lambda(v_\alpha)}{(1-\alpha)M}
+O_M(M^{-\frac{3}{2}}).
\end{eqnarray}
Combining \eqref{eq.b.17}, \eqref{eq.b.18}, and \eqref{eq.b.19},
\eqref{eq.b.13} holds and Lemma \ref{lem.b.3} is proven.
\end{proof}

\textbf{Proof of Lemma \ref{lem.b.4}.}
\begin{proof}
Let us first establish \eqref{eq.b.20}.
For simplicity, let us use $G(\cdot)$ and $\widetilde{G}_M(\cdot)$ to denote the inverse functions of $F(\cdot)$ and $\widetilde{F}_M(\cdot)$, respectively. Furthermore, denote $U(\theta)=\widetilde{F}_M(\widehat{H}_M(\theta))$. Clearly,
$\widehat{H}_M(\theta)=\widetilde{G}_M(U(\theta))$ and $\breve{v}^{M}_\alpha=\widetilde{G}_M(\alpha)$. It is easy to see that $U(\theta)$ is uniformly distributed over $[0,1]$. Moreover, from the relationship $\widehat{H}_M(\theta_{(1)})<\cdot\cdot\cdot<
\widehat{H}_M(\theta_{(N)})$, we know that $U(\theta_{(1)})<\cdot\cdot\cdot U(\theta_{(N)})$ is the corresponding order statistics of $N$ i.i.d. uniformly distributed random variables. Furthermore, let us use $\widehat{F}^N_u(\cdot)$ to denote the sample c.d.f. induced by $U(\theta_1),...,U(\theta_N)$. That is
$$
\widehat{F}^N_u(t)=\frac{1}{N}\sum\limits_{i=1}^N \mathds{1}\{
U(\theta_i)\le t\}.
$$
By Lemma \ref{lem.a.1}, we know that
\begin{eqnarray}\label{eq.b.22}
U(\theta_{(\alpha N)})-\alpha=\left(\alpha-\frac{1}{N}\sum\limits_{i=1}^{N}
\mathds{1}\{U(\theta_i)\le \alpha\}
\right)+O_{a.s.}(N^{-3/4}(\log N)^{3/4}).
\end{eqnarray}
Furthermore, by Taylor expansion,
\begin{eqnarray*}\label{eq.b.23}
\widetilde{v}^{N,M}_\alpha&=&\widehat{H}_M(\theta_{(\alpha N)})=
\widetilde{G}_M(U(\theta_{(\alpha N)}))\\
&=&\widetilde{G}_M(\alpha)+(U(\theta_{(\alpha N)})-\alpha)
\widetilde{G}_M^\prime(\alpha)+\frac{(U(\theta_{(\alpha N)})-\alpha)^2}{2} \widetilde{G}_M^{\prime\prime}(u)\\
&=&\breve{v}^{M}_\alpha+\frac{1}{\widetilde{f}_M(\breve{v}^{M}_\alpha)}
(U(\theta_{(\alpha N)})-\alpha)+\left(-\frac{\widetilde{f}^\prime_M(\widetilde{G}_M(u))}
{\widetilde{f}^3_M(\widetilde{G}_M(u))}\right)\frac{(U(\theta_{(\alpha N)})-\alpha)^2}{2},
\end{eqnarray*}
where $u$ lives in between $U(\theta_{(\alpha N)})$ and $\alpha$, and we use the facts that $\widetilde{G}_M(\alpha)=\breve{v}^{M}_\alpha$, $\widetilde{G}_M^\prime(\alpha)=1/\widetilde{f}_M(\breve{v}^{M}_\alpha)$,
and $\widetilde{G}_M^{\prime\prime}(u)
=\widetilde{f}^\prime_M(\widetilde{G}_M(u))/
\widetilde{f}^3_M(\widetilde{G}_M(u))$.
Therefore,
\begin{eqnarray}\label{eq.b.23}
\widetilde{v}^{N,M}_\alpha-\breve{v}^{M}_\alpha=
\frac{1}{\widetilde{f}_M(\breve{v}^{M}_\alpha)}
(U(\theta_{(\alpha N)})-\alpha)+\left(-\frac{\widetilde{f}^\prime_M(\widetilde{G}_M(u))}
{\widetilde{f}^3_M(\widetilde{G}_M(u))}\right)\frac{(U(\theta_{(\alpha N)})-\alpha)^2}{2}.
\end{eqnarray}
On the other hand, by Lemma 2.5.4B in \cite{serfling2009approximation}, we have for sufficiently large $N$
\begin{equation*}\label{eq.b.24}
|U(\theta_{(\alpha N)})-\alpha|\le 2 N^{-\frac{1}{2}}\left(\log N\right)^{\frac{1}{2}}.
\end{equation*}
Combining with \eqref{eq.b.22} and \eqref{eq.b.23}, we have
\begin{eqnarray}\label{eq.b.25}
&&\widetilde{v}^{N,M}_\alpha-\breve{v}^{M}_\alpha=
\frac{1}{\widetilde{f}_M(\breve{v}^{M}_\alpha)}\left\{
\left(\frac{1}{N}\sum\limits_{i=1}^{N}
\mathds{1}\{U(\theta_i)\le \alpha\}-
\alpha\right)+O_{a.s.}(N^{-3/4}(\log N)^{3/4})
\right\}\nonumber\\
&=&\frac{1}{\widetilde{f}_M(\breve{v}^{M}_\alpha)}
\left(\frac{1}{N}\sum\limits_{i=1}^{N}
\mathds{1}\{U(\theta_i)\le \alpha\}-\alpha\right)+
\frac{1}{\widetilde{f}_M(\breve{v}^{M}_\alpha)}
O_{a.s.}(N^{-3/4}(\log N)^{3/4})\nonumber\\
&=&\frac{1}{\widetilde{f}_M(\breve{v}^M_\alpha)}
\left(\frac{1}{N}\sum\limits_{i=1}^{N}
\mathds{1}\{\widehat{H}_M(\theta_{i})\le \breve{v}^M_\alpha\}-\alpha\right)+
\frac{1}{\widetilde{f}_M(\breve{v}^{M}_\alpha)}
O_{a.s.}(N^{-3/4}(\log N)^{3/4}).
\end{eqnarray}
Notice that $\widetilde{f}_M(\breve{v}^{M}_\alpha)$ is strictly positive and $\widetilde{f}_M(\breve{v}^{M}_\alpha)\rightarrow f(v_\alpha)$ as $M\rightarrow \infty$. Therefore, $\sup\limits_{M} 1/\widetilde{f}_M(\breve{v}^{M}_\alpha)<\infty$. Thus, \eqref{eq.b.20} holds.

It remains to show \eqref{eq.b.21}. Notice that by definition
\begin{eqnarray*}
&&\widetilde{c}^{N,M}_\alpha- \breve{c}^{M}_\alpha=
\frac{1}{(1-\alpha) N}\sum\limits_{i=1}^{N} \widehat{H}_M(\theta_{i})\mathds{1}\{\widehat{H}_M(\theta_{i})\ge \widetilde{v}^{N,M}_\alpha\}- \breve{c}^{M}_\alpha\\
&=&\widetilde{v}^{N,M}_\alpha+ \frac{1}{(1-\alpha) N}\sum\limits_{i=1}^{N} \left(\widehat{H}_M(\theta_{i})-\widetilde{v}^{N,M}_\alpha\right)^{+}
- \breve{c}^{M}_\alpha\\
&=&\left(\frac{1}{N}\sum\limits_{i=1}^{N}
\left[\breve{v}^M_\alpha+\frac{1}{1-\alpha}
\left(\widehat{H}_M(\theta_{i})-\breve{v}^M_\alpha\right)^{+}\right]
-\breve{c}^{M}_\alpha\right)\\
&&+\left(\widetilde{v}^{N,M}_\alpha-\breve{v}^M_\alpha\right)
+\frac{1}{(1-\alpha) N}\sum\limits_{i=1}^{N}\left[
\left(\widehat{H}_M(\theta_{i})-\widetilde{v}^{N,M}_\alpha\right)^{+}-
\left(\widehat{H}_M(\theta_{i})-\breve{v}^M_\alpha\right)^{+}
\right]\\
&=&\left(\frac{1}{N}\sum\limits_{i=1}^{N}
\left[\breve{v}^M_\alpha+\frac{1}{1-\alpha}
\left(\widehat{H}_M(\theta_{i})-\breve{v}^M_\alpha\right)^{+}\right]
-\breve{c}^{M}_\alpha\right)+(\ast),
\end{eqnarray*}
where
\begin{eqnarray*}
(\ast):=\left(\widetilde{v}^{N,M}_\alpha-\breve{v}^M_\alpha\right)
+\frac{1}{(1-\alpha) N}\sum\limits_{i=1}^{N}\left[
\left(\widehat{H}_M(\theta_{i})-\widetilde{v}^{N,M}_\alpha\right)^{+}-
\left(\widehat{H}_M(\theta_{i})-\breve{v}^M_\alpha\right)^{+}
\right].
\end{eqnarray*}
We only need to show that $(\ast)$ is in the order of $O_{a.s.}(N^{-1}\log N)$ uniformly for all $M$. Note that the second term in $(\ast)$
\begin{eqnarray*}
&&\frac{1}{(1-\alpha) N}\sum\limits_{i=1}^{N}\left[
\left(\widehat{H}_M(\theta_{i})-\widetilde{v}^{N,M}_\alpha\right)^{+}-
\left(\widehat{H}_M(\theta_{i})-\breve{v}^M_\alpha\right)^{+}
\right]\\
&=&\frac{1}{(1-\alpha) N}\sum\limits_{i=1}^{N}\Bigl[
\left(\widehat{H}_M(\theta_{i})-\widetilde{v}^{N,M}_\alpha\right)
\mathds{1}\{\widehat{H}_M(\theta_{i})\ge \widetilde{v}^{N,M}_\alpha\}\\
&&-\left(\widehat{H}_M(\theta_{i})-\breve{v}^M_\alpha\right)
\mathds{1}\{\widehat{H}_M(\theta_{i})\ge \breve{v}^{M}_\alpha\}
\Bigr]\\
&=&\frac{1}{(1-\alpha) N}\sum\limits_{i=1}^{N}\left[
\left(\breve{v}^{M}_\alpha-\widetilde{v}^{N,M}_\alpha\right)
\mathds{1}\{\widehat{H}_M(\theta_{i})\ge \widetilde{v}^{N,M}_\alpha\}\right]\\
&&+\frac{1}{(1-\alpha) N}\sum\limits_{i=1}^{N}\left(\widehat{H}_M(\theta_{i})-
\breve{v}^{M}_\alpha\right)\left[
\mathds{1}\{\widehat{H}_M(\theta_{i})\ge \widetilde{v}^{N,M}_\alpha\}-\mathds{1}\{\widehat{H}_M(\theta_{i})\ge \breve{v}^{M}_\alpha\}
\right]\\
&=&\frac{1}{(1-\alpha)N}\left(\breve{v}^{M}_\alpha-
\widetilde{v}^{N,M}_\alpha\right)+\frac{1}{(1-\alpha) N}\sum\limits_{i=1}^{N}\left[
\left(\widetilde{v}^{N,M}_\alpha-\breve{v}^{M}_\alpha\right)
\mathds{1}\{\widehat{H}_M(\theta_{i})\le \widetilde{v}^{N,M}_\alpha\}\right]\\
&&+\frac{1}{(1-\alpha) N}\sum\limits_{i=1}^{N}\left(\widehat{H}_M(\theta_{i})-
\breve{v}^{M}_\alpha\right)\left[
\mathds{1}\{\widehat{H}_M(\theta_{i})\le \breve{v}^{M}_\alpha\}-\mathds{1}\{\widehat{H}_M(\theta_{i})\le \widetilde{v}^{N,M}_\alpha\}
\right]\\
&=&\frac{1}{(1-\alpha)N}\left(\breve{v}^{M}_\alpha-
\widetilde{v}^{N,M}_\alpha\right)+\frac{1}{(1-\alpha) N}\sum\limits_{i=1}^{N}\left[
\left(\widetilde{v}^{N,M}_\alpha-\breve{v}^{M}_\alpha\right)
\mathds{1}\{\widehat{H}_M(\theta_{i})\le \widetilde{v}^{N,M}_\alpha\}\right]+(\ast\ast\ast),
\end{eqnarray*}
where
$$
(\ast\ast\ast)=\frac{1}{(1-\alpha) N}\sum\limits_{i=1}^{N}\left(\widehat{H}_M(\theta_{i})-
\breve{v}^{M}_\alpha\right)\left[
\mathds{1}\{\widehat{H}_M(\theta_{i})\le \breve{v}^{M}_\alpha\}-\mathds{1}\{\widehat{H}_M(\theta_{i})\le \widetilde{v}^{N,M}_\alpha\}
\right].
$$
Further note that
\begin{eqnarray*}
&&\left(\widetilde{v}^{N,M}_\alpha-\breve{v}^M_\alpha\right)
+\frac{1}{(1-\alpha)N}\left(\breve{v}^{M}_\alpha-
\widetilde{v}^{N,M}_\alpha\right)\\
&&+\frac{1}{(1-\alpha) N}\sum\limits_{i=1}^{N}\left[
\left(\widetilde{v}^{N,M}_\alpha-\breve{v}^{M}_\alpha\right)
\mathds{1}\{\widehat{H}_M(\theta_{i})\le \widetilde{v}^{N,M}_\alpha\}\right]\\
&=&\frac{1}{(1-\alpha)}
\left(\widetilde{v}^{N,M}_\alpha-\breve{v}^M_\alpha\right)
\left[\frac{1}{N}\sum\limits_{i=1}^{N}
\mathds{1}\{\widehat{H}_M(\theta_{i})\le \widetilde{v}^{N,M}_\alpha\} -\alpha\right]\\
&=&\frac{1}{(1-\alpha)}
\left(\widetilde{v}^{N,M}_\alpha-\breve{v}^M_\alpha\right)
\left[\frac{1}{N}\sum\limits_{i=1}^{N}
\mathds{1}\{U(\theta_{i})\le U(\theta_{(\alpha N)})\} -\alpha\right]\\
&=&\frac{1}{(1-\alpha)}
\left(\widetilde{v}^{N,M}_\alpha-\breve{v}^M_\alpha\right)\left(
\widehat{F}^N_u(U(\theta_{(\alpha N)}))-\alpha\right)
\overset{\triangle}=(\ast\ast).
\end{eqnarray*}
Note that $(\ast)=(\ast\ast)+(\ast\ast\ast)$, we only need to show that $(\ast\ast)$ and $(\ast\ast\ast)$ both are in the order of $O_{a.s.}(N^{-1}\log N)$ uniformly for all $M$.

By Lemma 2.5.4B in \cite{serfling2009approximation}, we know that for sufficiently large $N$ (can be verified this is uniform for all $M$, as in \eqref{eq.b.25})
\begin{eqnarray}\label{eq.b.26}
|\widetilde{v}^{N,M}_\alpha-\breve{v}^M_\alpha|\le
\frac{2}{\widetilde{f}_M(\breve{v}^M_\alpha)}N^{-\frac{1}{2}}\left(\log N\right)^{\frac{1}{2}}.
\end{eqnarray}
Moreover, by applying Theorem 2.5.1 and Lemma 2.5.4B in \cite{serfling2009approximation} on $U(\theta)$, we have for sufficiently large $N$
\begin{eqnarray}\label{eq.b.27}
|\widehat{F}^N_u(\alpha)-\alpha|=2N^{-\frac{1}{2}}\left(\log N\right)^{\frac{1}{2}}+O_{a.s.}(N^{-3/4}(\log N)^{3/4}).
\end{eqnarray}
Applying Lemma 2.5.4B and Lemma 2.5.4E (with $c_0=2$, $q=1/2$) in \cite{serfling2009approximation} on $U(\theta)$,
we have for sufficiently large $N$
\begin{eqnarray}\label{eq.b.28}
|\widehat{F}^N_u(U(\theta_{(\alpha N)}))-\widehat{F}^N_u(\alpha)|=2N^{-\frac{1}{2}}\left(\log N\right)^{\frac{1}{2}}+O_{a.s.}(N^{-3/4}(\log N)^{3/4}).
\end{eqnarray}
Combining \eqref{eq.b.27} and \eqref{eq.b.28}, we have for sufficiently large $N$
\begin{eqnarray*}
|\widehat{F}^N_u(U(\theta_{(\alpha N)}))-\alpha|\le
4N^{-\frac{1}{2}}\left(\log N\right)^{\frac{1}{2}}+O_{a.s.}(N^{-3/4}(\log N)^{3/4}).
\end{eqnarray*}
Combining with \eqref{eq.b.26}, we have for sufficiently large $N$ (uniform for all $M$)
\begin{eqnarray*}\label{eq.b.29}
(\ast\ast)=\frac{8}{\widetilde{f}_M(\breve{v}^M_\alpha)}
\left(N^{-1}\left(\log N\right)+O_{a.s.}(N^{-5/4}(\log N)^{5/4})\right)
\end{eqnarray*}
In view of the fact that $\sup\limits_{M} 1/\widetilde{f}_M(\breve{v}^{M}_\alpha)<\infty$, we have $(\ast\ast)$ in the order of $O_{a.s.}(N^{-1}\log N)$ uniformly for all $M$. What is left is show $(\ast\ast\ast)$ is also in the order of $O_{a.s.}(N^{-1}\log N)$ uniformly for all $M$.
\begin{eqnarray*}
|(\ast\ast\ast)|&=&\left|\frac{1}{(1-\alpha)N}
\sum\limits_{i=1}^{N}\left(\widehat{H}_M(\theta_{i})-
\breve{v}^{M}_\alpha\right)\left[
\mathds{1}\{\widehat{H}_M(\theta_{i})\le \breve{v}^{M}_\alpha\}-\mathds{1}\{\widehat{H}_M(\theta_{i})\le \widetilde{v}^{N,M}_\alpha\}
\right]\right|\\
&\le & \frac{1}{(1-\alpha)}
\left|\widetilde{v}^{N,M}_\alpha-\breve{v}^M_\alpha \right|
\left|\frac{1}{N}
\sum\limits_{i=1}^{N}\mathds{1}\{\widehat{H}_M(\theta_{i})\le \breve{v}^{M}_\alpha\}-\frac{1}{N}
\sum\limits_{i=1}^{N}\mathds{1}\{\widehat{H}_M(\theta_{i})\le \widetilde{v}^{N,M}_\alpha\}\right|\\
&=& \frac{1}{(1-\alpha)}
\left|\widetilde{v}^{N,M}_\alpha-\breve{v}^M_\alpha \right|
\left|\frac{1}{N}
\sum\limits_{i=1}^{N}\mathds{1}\{U(\theta_{i})\le \alpha\}-\frac{1}{N}
\sum\limits_{i=1}^{N}\mathds{1}\{U(\theta_{i})\le U(\theta_{(\alpha N)})\}\right|\\
&=&\frac{1}{(1-\alpha)}
\left|\widetilde{v}^{N,M}_\alpha-\breve{v}^M_\alpha \right|
\left|\widehat{F}^N_u(U(\theta_{(\alpha N)}))-\widehat{F}^N_u(\alpha)
\right|.
\end{eqnarray*}
By \eqref{eq.b.26} and \eqref{eq.b.28}, we have for sufficiently large $N$ (uniform for all $M$)
\begin{eqnarray*}\label{eq.b.30}
(\ast\ast)&\le& \frac{1}{(1-\alpha)}
\left|\widetilde{v}^{N,M}_\alpha-\breve{v}^M_\alpha \right|
\left|\widehat{F}^N_u(U(\theta_{(\alpha N)}))-\widehat{F}^N_u(\alpha)
\right|\nonumber \\
&=&\frac{4}{\widetilde{f}_M(\breve{v}^M_\alpha)}
\left(N^{-1}\left(\log N\right)+O_{a.s.}(N^{-5/4}(\log N)^{5/4})\right).
\end{eqnarray*}
Again, in view of the fact that $\sup\limits_{M} 1/\widetilde{f}_M(\breve{v}^{M}_\alpha)<\infty$, we have $(\ast\ast\ast)$ in the order of $O_{a.s.}(N^{-1}\log N)$ uniformly for all $M$.
\end{proof}

\end{document}